\documentclass[preprint]{ptephy_v1}%%%%%% to generate preprint number

\preprintnumber{} %%% %%% Insert preprint number here

\usepackage{cases}
\usepackage{amsthm}
\newtheorem*{theorem*}{Theorem}
\newtheorem*{definition*}{Definition}
\newtheorem{lemma}{Lemma}
\usepackage{amsmath}
\newtheorem{theorem}{Theorem}

\newtheorem{corollary}{Corollary}

\begin{document}

\title{Violation of cosmic censorship in the gravitational collapse of a dust cloud in five dimension}

%%%% To generate auto affiliation numbers please use \author{}\affil{} command

\author{Ryosuke Mizuno}
\affil{Department of Physics, Kyoto University, Kyoto 606-8502, Japan\email{mizuno@tap.scphys.kyoto-u.ac.jp}}

\author{Seiju Ohashi}
\affil{Theory Center, Institute of Particles and Nuclear Studies, KEK, Tsukuba, Ibaraki, 305-0801, Japan}

\author{Tetsuya Shiromizu}
\affil{Department of Mathematics, Nagoya University, Nagoya 464-8602, Japan}
\affil{Kobayashi-Maskawa Institute, Nagoya University,Nagoya 464-8602, Japan}

%%% To include the collaborator name... Please use the command "\collaborator"
%%% For example: \collaborator{ATLAS Collaboration}

\begin{abstract}
We analyze the null geodesic equations in five dimensional spherically symmetric spacetime 
with collapsing inhomogeneous dust cloud. By using a new method, we prove the existence and 
non-existence of solutions to null geodesic equation emanating from central singularity for 
smooth initial distribution of dust. Moreover, we also show that the null geodesics can 
extend to null infinity in a certain case, which imply the violation of cosmic censorship conjecture. 
\end{abstract}

\subjectindex{}

\maketitle

\section{Introduction}

Black hole spacetime is one of the most fascinating objects in gravitational theory. In particular, 
it is quite interesting that they contain singularities inside of their event horizons. At singularities 
the spacetime curvature often diverges. The physics breaks down at singularities because 
gravitational theory is described in terms of curvature. 

Observers outside of black holes can not see such breakdown because no information can come out from black holes at least 
in classical level. However, black hole singularities may cause serious effects on the observer inside of black hole and/or 
the final fate of black hole evapolation due to the Hawking radiation.  Here we 
define the singularity which is visible for observer inside of black hole, as ``locally naked singularity". 
If a singularity is not wrapped by horizon, it may cause serious effect on physics.  We define that as 
``globally naked singularity".

In order to certify the predictability of physics, it is very important to ask whether they can be naked 
or not. It is usually supposed that no naked singularity will appear in physical situation. In this context, 
Penrose proposed so called cosmic censorship conjecture(CCC) \cite{Penrose:1969pc}. 
More precisely, there are two types of CCC, i.e. strong CCC and weak CCC. Strong CCC states that there are no locally 
naked singularities form during gravitational collapse. And weak CCC does that there are no globally singularities 
form during collapse.

The cosmic censorship conjecture is assumed to prove several key theorems on black hole spacetime, 
such as the event horizon topology theorem, uniqueness theorem and so on (see \cite{Hawking:1973uf} for the details). 
The conjecture have been investigated by many authors in variety of setups, but it is still controversial.
In \cite{Oppenheimer:1939ue}, Oppenheimer and Snyder considered a spherical collapse of homogeneous 
pressureless fluid, which is called as dust. And they found that no naked singularity form. 
But in subsequent works, for example  \cite{Yodzis:1973, Eardley:1978tr, Newman:1985gt, Joshi:1993zg, Singh:1994tb, Jhingan:1996jb,Joshi:2008zz}, 
many works reported the violation of strong CCC in several situations. In \cite{Christodoulou:1984mz}, 
Christodoulou examined the global nakedness of singularity in four dimensional spacetime. He considered 
the spherical collapse of inhomogeneous dust and proved that singularity can be globally 
naked in some situations, i.e. the violation of weak CCC occurs in general.  

Motivated by fundamental theory such as string theory, those works have been extended to higher dimensional 
spacetime \cite{Ghosh:2001fb, Goswami:2004gy, Goswami:2006ph}. It was shown that the strong CCC always holds 
for the spherical collapse of dust cloud with smooth initial data in higher dimensional spacetimes than five, 
i.e. any observer can not see singularity formed. On the other hand, in general, it was also shown that 
the strong CCC does not hold in five dimensional spacetimes.

As far as we know, there is no work on a global analysis of collapsing spacetime in {\it five} dimensions. It is 
natural to ask whether the weak CCC is actually violated in five dimension. And, if violated, it is important to 
clarify in what conditions naked singularities form. In this paper we focus on the analysis of five 
dimensional inhomogeneous spherically symmetric dust collapse and give a new method to examine the nakedness 
of singularity. To do so we have to know whether a causal geodesic emanating from the singularity exists or not. 
So, we give a method to investigate the existence of solution of null geodesic equation. Furthermore, we 
examine the spacetime structure in detail, and give the necessary and sufficient condition for naked singularity 
formation. We also examine the condition that the singularity is not only locally naked but also globally naked. 
Finally, we will see the dependence of the globally nakedness of the singularity on initial density distribution. 

The organization of this paper is as follows. In Sec. 2, we present the settings and the fundamental nature of 
the inhomogeneous spherically symmetric dust collapse in five dimensional spacetime. Then, we derive the differential 
equation for null geodesic in terms of dimensionless quantities. In Sec. 3, we examine the existence condition of 
the solution to the differential equation for null geodesic. By virtue of the Schauder fixed-point theorem \cite{Schauder:1930}, 
we show that a solution of the differential equation for null geodesic exist near the singularity. It means that the singularity 
can be at least locally naked. In Sec. 4, we analyze the spacetime structure around the singularity. We identify 
the earliest null geodesic emanating from the central singularity, and give the necessary and sufficient condition 
for the singularity to be naked. In Sec. 5, we consider the globally nakedness of the singularity. We will show that 
a class of initial density distribution leads to globally naked singularity.

\section{Five dimensional Lema\^itre-Tolman-Bondi(LTB) spacetime and the equation of null line}

We consider spherically symmetric dust collapse in five dimension. It is known as LTB solution in higher dimension. 
In the comoving coordinate of dust, the metric of this spacetime is written as

\begin{equation}
\label{metric}
ds^2 = -dt^2 + e^{2\omega(t,r)}dr^2 + R(t,r)^2 d\Omega^2,
\end{equation}
where $R$ is the area radius of the $r = const.$ 3-sphere and we set $r$ by $R(0,r)=r$. Then, in this coordinate, 
the Einstein equations become
\begin{equation}
\rho(t,r) = \frac{3}{2} \frac{M^{\prime}(r)}{R^3R^{\prime}},
\label{EQ1}
\end{equation}
\begin{equation}
\dot{R}^2 = \frac{M(r)}{R^2} + E(r),
\label{EQ2}
\end{equation}
\begin{equation}
e^{2\omega} = \frac{R^{\prime2}}{1+E(r)},
\label{EQ3}
\end{equation}
where $\rho(t,r)$ is the energy density of dust, which is proportional to the Ricci scalar, $M(r)$ is an arbitrary 
$C^1$ function of $r$ and $E(r)$ is arbitrary functions of $r$. dot `` $\dot{}$ " and prime `` ${}^{\prime}$ " mean 
the partial derivatives with respect to $t$ and $r$,  respectively.  $M(r)$ corresponds to the total mass in the region 
surrounded by the $r = const.$ surface. Actually $M(r)$ is proportional to the Misner-Sharp quasi-local mass 
\cite{Misner:1964je}. And $E(r)$ is the initial energy of dust shell. 

As we mentioned in the introduction, the point where spacetime curvature diverges is singularity. From equation (\ref{EQ1}), 
we can find the two types of singularities, i.e. $R = 0$ and $R^{\prime} = 0$. Singularity at $R=0$ and $R^{\prime}=0$ 
is called shell focusing singularity and shell crossing singularity, respectively. If we introduce the pressure to the fluid, 
the shell crossing singularity may disappear. Therefore they are regarded as unphysical singularity. Throughout this paper, 
we only consider the shell focusing singularity, that is, we assume $R^{\prime} > 0$.  In addition, we also assume that the 
initial velocity of the shells at $t = 0$ is zero, that is,
\begin{equation}\label{DofE}
\dot{R}^2(0,r) = \frac{M(r)}{r^2} + E(r) = 0.
\end{equation}
Now we can solve the equation (\ref{EQ2}) as
\begin{equation}\label{T_R}
t(R,r) = - \int_r^R dR \frac{1}{\sqrt{M(r)\left( \frac{1}{R^2} - \frac{1}{r^2}\right)}} = \frac{r^2}{\sqrt{M(r)}}\sqrt{ 1 - \frac{R^2}{r^2}}.
\end{equation}
From this equation and (\ref{EQ1}), we see that the singularity appears at
\begin{equation}\label{TS}
t_S(r) = \frac{r^2}{\sqrt{M(r)}}.
\end{equation}
Using equations (\ref{metric}), (\ref{EQ2}) and (\ref{EQ3}), we can compute the expansion of outgoing null geodesics on 
$r = const.$ surface $\Theta(r)$ as
\begin{equation}\label{EXP}
\begin{split}
\Theta(r) \propto \frac{dR(t,r_{null}(t))}{dt} &= \dot{R} + R^{\prime} \frac{dr_{null}}{dt}  \\
&= \dot{R} + R^{\prime} e^{-\omega} \\
&= \sqrt{ 1 - \frac{M(r)}{r^2} } -  \sqrt{  \frac{M(r)}{R^2} - \frac{M(r)}{r^2} }.
\end{split}
\end{equation}
In the second equality, we used the equation $-dt^2 + e^{2\omega}dr^2=0$ which holds for null geodesics along the outer 
radial direction. This equation implies that the apparent horizon ($\Theta =0$) is located at $R = \sqrt{M(r)}$. 
Since we are interested in a naked singularity formation during the dust collapse, we assume that  there is no apparent 
horizon initially. Then, equation (\ref{EXP}) with the setting of $R(0,r) = r$ tells us that
\begin{equation}
\label{Mcondition}
1 - \frac{M(r)}{r^2} > 0
\end{equation}
is required for all $r$ in this coordinate patch. From equations (\ref{metric}), (\ref{EQ3}) and (\ref{DofE}), this is equivalent to the condition that $r$ 
is spacelike coordinate. In addition, from equation (\ref{T_R}), we also see that the apparent horizon appears at
\begin{equation}\label{TAH}
t_{AH}(r) =  \frac{r^2}{\sqrt{M(r)}}\sqrt{ 1 - \frac{M(r)}{r^2}}.
\end{equation}

From equations (\ref{TS}) and (\ref{TAH}), if $r \neq 0$, it is easy to see that
\begin{equation}
t_{AH} (r) < t_S(r)
\end{equation}
holds. This means that the apparent horizon appears before the singularity does and the singularity is surrounded 
by the apparent horizon. So only the singularity at $r = 0$ could be naked and another singularities  are covered by 
the event horizon \cite{Hawking:1973uf}. In order to examine the possibility of the occurrence of naked singularity, we assume that the 
singularity at $r = 0$ appears within non-zero finite time $t_S(0)$, that is,
\begin{equation}
t_S(0) = \lim_{r \to 0} t_S(r)  =\lim_{r \to 0} \frac{r^2}{\sqrt{M(r)}}.
\end{equation}
Therefore, using a $C^0$ function $A(r)$ with $A(0) \neq 0$, we can write $M(r)$ as
\begin{equation}
\label{DA}
M(r) = A(r)r^4.
\end{equation}
In the above, $A(r)$ corresponds to the mean density for the region surrounded by the $r = const.$ shell. From equations (\ref{EQ1}) 
and (\ref{DA}), we have
\begin{equation}\label{Arho}
A(r) = \frac{2}{3}\int_0^1dv v^3\rho(0,vr).
\end{equation}

Moreover, we assume that the initial density $\rho(0,r)$ is a $C^\infty$ function of compact support and monotonically decreases 
with respect to $r$, and $\rho^{\prime}(0,r)$ is continuous at $r = 0$ on the initial time slice $(\rho^{\prime}(0,r) \leq 0$ , $\rho^{\prime}(0,0) = 0$). 
Then, we see that $A(r)$ satisfies\footnote{In this paper, $f(x)=O(x^a)$ means that the absolute value of function $f(x)$ is bounded by constant 
times $x^a$ as $x \to 0$, i.e there exist a positive $x_0$ and $c$ in $\mathbb{R}^+$ such that $|f(x)| \leq cx^a$ holds for all $x \in [0,x_0]$.}
\begin{equation}
\label{Acondition1}
A(r) = \alpha - \frac{\beta}{2}r^2 + O(r^3)
\end{equation}
near $r = 0$ and
\begin{equation}
\label{Acondition2}
A^{\prime}(r) \leq 0,
\end{equation}
where $\alpha,\beta\in \mathbb{R}$ are parameters satisfying $\alpha > 0, \beta\geq 0$.
Using $A(r)$, $t_{S}(r)$ and $t_{AH}(r)$ are rewritten as
\begin{equation}\label{tS}
t_S (r)=\sqrt{ \frac{1}{A(r)}},
\end{equation}
\begin{equation}\label{tAH}
t_{AH} (r)= \sqrt{\frac{1}{A(r)}( 1 - A(r)r^2)}.
\end{equation}

If $\beta = 0$ in equation (\ref{Acondition1}), we can immediately show that central singularity must be covered by apparent horizon. 
\begin{theorem}\label{T1}
If $\beta = 0$, the strong cosmic censorship holds. 
\end{theorem}

\begin{proof}\label{T1}
$\beta = 0$ implies
\begin{equation}
A(r) - \alpha = O(r^3).
\end{equation}
Then,
\begin{equation}
\begin{split}
t_{AH}(r) - t_S(0) &=  \sqrt{\frac{1}{A(r)}( 1 - A(r)r^2)} - \sqrt{ \frac{1}{\alpha}} \\
&= \frac{1 } {\sqrt{\alpha A(r)}} \left( \sqrt{\alpha}- \frac{\sqrt{\alpha^3}}{2}r^2 - \sqrt{A(r)} + O(r^3) \right)\\
&= -\frac{\sqrt{\alpha}}{2}r^2 + O(r^3) .
\end{split}
\end{equation}
Thus, there exists $r_0 \in \mathbb{R}$ such that
\begin{equation}
t_{AH}(r) - t_S(0) < 0
\end{equation}
holes for arbitrary $r \in\left[0, r_0\right]$. Therefore, there exists apparent horizon around the past of central singularity. So null geodesics can not 
emanate from the singularity and the strong cosmic censorship holds.
\end{proof}

This fact is already known in Refs. \cite{ Ghosh:2001fb,Goswami:2004gy, Goswami:2006ph}. Accordingly, in order to figure out 
the condition for the naked singularity formation, we suppose $\beta > 0$ in the following discussion. Let us consider 
the future directed null geodesics along the outer radial direction. Because of spherical symmetry, the differential 
equation for future directed null geodesics along the outer radial direction is given by
\begin{equation}\label{nulleq_0}
\begin{split}
\frac{dt}{dr} &= e^{\omega} \\
&=\frac{R^{\prime}}{\sqrt{1-A(r)r^2}}\\
&= \frac{1}{\sqrt{(1-A(r)r^2)(1-A(r)t^2)}}\left(    1 - A(r)t^2 - \frac{A^{\prime}(r)}{2}rt^2     \right).
\end{split}
\end{equation}
We used equation (\ref{metric}) in the first line, equations (\ref{EQ3}), (\ref{DofE}) and (\ref{DA}) in the second 
line and the explicit expression of $R^{\prime}$ derived from equation (\ref{T_R}) in the third one. 

Now, we introduce the dimensionless functions and parameters to write equation (\ref{nulleq_0}) in the dimensionless 
form. Let $a(x)$ to be a function in $C^\infty[0,\infty)$ such that
\begin{numcases}
{}
a(0) = 1 &\label{Cofa1} \\
\frac {da(0)}{dx} = 0 & \label{Cofa2}\\
\frac {d^2 a(0)}{dx^2} = -1 & \label{Cofa3}\\
\frac {da(x)}{dx} \leq 0 & \label{monoa}\\
a(x) = \frac{m}{x^4}  & $(x \geq l)$. \label{vaca}
\end{numcases}
In the above, $m$ and $l$ are dimensionless parameters which are related to the total mass of the dust and the dust cloud radius, 
respectively. These conditions imply that $a(x)$ is written as
\begin{equation}\label{EXPofa}
a(x) = 1 - \frac{x^2}{2} + O(x^3).
\end{equation}
Using $a(x)$, we write $A(r)$ as
\begin{equation}
A(r) = \alpha a \left(\sqrt {\frac{\beta}{\alpha}}r \right). \label{defa}
\end{equation}
We can easily show that this $A(r)$ satisfies the conditions (\ref{Acondition1}) and (\ref{Acondition2}). In the above, 
$\alpha$ and $\beta$ are non-zero positive parameters satisfying 
\begin{equation}
 \label{eta_0}
\max_{0 \leq x \leq l} a(x)x^2 \equiv \eta < \frac{\beta}{\alpha^2}.
\end{equation}
This condition comes from equation (\ref{Mcondition}). In this way, $A(r)$ is parameterized by the 2 parameters $\alpha$, $\beta$. From equation (\ref{Arho}), the initial density $\rho(0,r)$ is also parameterized as 
\begin{equation}
 \label{DPrho}
\rho(0,r) = \frac{3\alpha}{2r^3}\frac{d}{dr}\left( r^4 a \left(\sqrt {\frac{\beta}{\alpha}}r \right) \right).
\end{equation}
In the following, the function $a(x)$ is given so that it satisfies equations (\ref{Cofa1})-(\ref{vaca}) and 
the initial density distribution is parameterized by $\alpha$ and $\beta$. Using $\alpha$, $\beta$ and $a(x)$, equation (\ref{nulleq_0}) is rewrriten as
\begin{equation}
\frac{dt}{dr} = \frac{1}{\sqrt{\Bigl(1-\alpha a(\sqrt {\frac{\beta}{\alpha}}r)r^2\Bigr)\Bigl(1-\alpha a(\sqrt {\frac{\beta}{\alpha}}r)t^2\Bigr)}}
\left(    1 - \alpha a \Bigl(\sqrt {\frac{\beta}{\alpha}}r \Bigr)t^2 - \frac{\alpha rt^2}{2}\frac{d}{dr}a \Bigl(\sqrt {\frac{\beta}{\alpha}}r \Bigr)    \right).
\end{equation}
Moreover, using dimensionless coordinates
\begin{equation}\label{Dfx}
x \equiv \sqrt{\frac{\beta}{\alpha}}r
\end{equation}
and
\begin{equation}
\label{Dzeta}
\zeta \equiv \sqrt{\alpha}(t-t_S(0)) = \sqrt{\alpha}\left(t- \frac{1}{\sqrt{\alpha}}\right),
\end{equation}
we obtain the dimensionless equation for the null line
\begin{equation}\label{nulleq_1}
\frac{d\zeta}{dx} = \frac{1}{\sqrt{\gamma \left( 1-\frac{a(x)x^2}{\gamma} \right) 
\left( 1-a(x)(\zeta+1)^2 \right) }}\left(    1 - a(x)(\zeta + 1)^2 - \frac{x}{2}(\zeta+1)^2\frac{d}{dx}a(x)    \right),
\end{equation}
where
\begin{equation}
\label{Dgamma}
\gamma \equiv \frac{\beta}{\alpha ^2}.
\end{equation}
Equation (\ref{eta_0}) is also rewriten as
\begin{equation}
\label{eta}
\gamma > \eta.
\end{equation}
Here note that the right-hand side of equation (\ref{nulleq_1}) is not Lipschitz continuous function 
in the region that contains $x = \zeta = 0$. If this equation has a solution which starts from 
$x = \zeta = 0$, then, at least, singularity is locally naked. Moreover, if the solution could extend to $x \to \infty$, 
the singularity would be visible at null infinity, that is, it is globally naked. From now on we will ask if this 
differential equation has a solution. For the convenience, we define the dimensionless coordinate $\theta$ as
\begin{equation}\label{Dftheta}
\theta x^2 \equiv \zeta.
\end{equation}
Then equation (\ref{nulleq_1}) becomes
\begin{equation}\label{nulleq_2}
\frac{d\theta}{dx} + \frac{2\theta}{x}= \frac{1}{x^2\sqrt{\gamma \left( 1-\frac{a(x)x^2}{\gamma} \right) \left( 1-a(x)(\theta x^2+1)^2 \right) }}\left(    1 - a(x)(\theta x^2 + 1)^2 - \frac{x}{2}(\theta x^2+1)^2\frac{d}{dx}a(x)    \right).
\end{equation}
In the current expression, from equations (\ref{tS}), (\ref{tAH}) and the relation $\theta = \frac{\sqrt{\alpha}t - 1}{x^2}$, 
we see that the singularity and the apparent horizon is located at
\begin{equation}\label{locationofS}
\theta_S (x)= \frac{1}{x^2} \left(\frac{1}{\sqrt{a(x)}} - 1\right)
\end{equation}
and
\begin{equation}\label{locationofAH}
\theta_{AH}(x;\gamma)= \frac{1}{x^2} \left(\frac{1}{\sqrt{a(x)}} \sqrt{1 - \frac{a(x)x^2}{\gamma}} - 1\right),
\end{equation}
respectively. And, 
\begin{equation}
\theta_S (0)= \frac{1}{4},
\end{equation}
\begin{equation}
\theta_{AH} (0;\gamma)= \frac{1}{4} - \frac{1}{2\gamma}.
\end{equation}

Note that, from the coordinate transformations (\ref{Dzeta}) and (\ref{Dftheta}), we see that the region $x = 0$ is 
singular for arbitrary $\theta$. But any future directed curve does not emanate from the central singularity 
located in $\theta < 0$ because any point in the region satisfying $x \neq 0$ and $\theta < 0$ is not in the future time slice of the central singularity. Then, for our purpose, we focus on the central singularity located in 
$\theta \geq 0$.

If the differential equation (\ref{nulleq_2}) has a $C^1$ solution $\theta(x)$ for $x \in [0,l]$, the multiplication 
of $x$ with the right-hand side of equation (\ref{nulleq_1}) for its solution behaves around $x=0$ as
\begin{equation}\label{limf}
\begin{split}
 \lim_{x \to 0} \frac{1}{x\sqrt{\gamma \left( 1-\frac{a(x)x^2}{\gamma} \right) \left( 1-a(x)(\theta x^2+1)^2 \right) }} & \left(    1 - a(x)(\theta x^2 + 1)^2 - \frac{x}{2}(\theta x^2+1)^2\frac{d}{dx}a(x)    \right) \\
&=  \frac{ 1- 2\theta(0)}{ \sqrt {\gamma \left( \frac{1}{2} - 2\theta(0) \right)} }.  \\
\end{split}
\end{equation}
Thus, if 
\begin{equation}
2\theta(0) =  \frac{ 1- 2\theta(0)}{ \sqrt {\gamma \left( \frac{1}{2} - 2\theta(0) \right)} }
\end{equation}
holds, then the first-order pole of equation (\ref{nulleq_2}) is cancelled out and does not appear. As with Christodoulou \cite{Christodoulou:1984mz}, 
let us introduce a real number $\lambda$ satisfying
\begin{equation}
\label{lambda}
2\lambda = \frac{ 1- 2\lambda}{ \sqrt {\gamma \left( \frac{1}{2} - 2\lambda \right)} }. \
\end{equation}
Since $\gamma$ is a positive real number (see below equation (\ref{defa}) and equation (\ref{Dgamma})), we have
\begin{equation}
0 < \lambda < \frac{1}{4}.
\end{equation}
In addition, since
\begin{equation}
\frac{d\sqrt {\gamma(\lambda)}}{d\lambda} = \frac{d}{d\lambda} \left( \frac{ 1- 2\lambda}{ 2\lambda \sqrt { \left( \frac{1}{2} - 2\lambda \right)} } \right)
= \frac{-4\lambda^2 + 6\lambda -1}{4\lambda^2 \left( \frac{1}{2} - 2\lambda\right)^{\frac{3}{2}}},
\end{equation}
we see
\begin{equation}
\min_{0 < \lambda < \frac{1}{4}} \sqrt{\gamma(\lambda)}  = \sqrt{\gamma \left( \lambda_{M} \right)} = \sqrt{11 + 5\sqrt{5}} \equiv \gamma_{min},
\end{equation}
where
\begin{equation}\label{lamdbapm}
\lambda_{M} \equiv \frac{3 - \sqrt{5}}{4},
\end{equation}
that is, $\gamma$ has the minimum at $\lambda = \lambda_{M}$. 
If $\gamma >\gamma_{min}$ holds, we have the two solution to equation (\ref{lambda}) for given $\gamma$, $\lambda_-(\gamma)$ and $\lambda_+(\gamma)$, 
satisfying $0 < \lambda_-(\gamma) < \lambda_{M}  <\lambda_+(\gamma)<\frac{1}{4}$. If $\gamma = \gamma_{min}$, equation (\ref{lambda}) has the single 
solution $\lambda_{\pm}(\gamma) = \lambda_{M}$. If $\gamma < \gamma_{min}$, equation (\ref{lambda}) has no solution. From equation (\ref{lambda}), 
it is easy to see that $\lambda_-(\gamma) $ and $\lambda_+(\gamma)$ satisfy $\lim_{\gamma \to \infty}  \lambda_-(\gamma)= 0$ and 
$\lim_{\gamma \to \infty} \lambda_+(\gamma) = \frac{1}{4}$.

If $\lambda$ satisfying equation (\ref{lambda}) exists, we rewrite equation (\ref{nulleq_2}) as
\begin{equation}\label{nulleq}
\begin{split}
\frac{d\theta}{dx} - 2 \frac{(\lambda-\theta)}{x} &= \frac{   1 - a(x)(\theta x^2 + 1)^2 - \frac{x}{2}(\theta x^2+1)^2\frac{d}{dx}a(x)    }{x^2\sqrt{\gamma \left( 1-\frac{a(x)x^2}{\gamma} \right) \left( 1-a(x)(\theta x^2+1)^2 \right) }} - \frac{2\lambda}{x}  \\
& \equiv \lambda f(x,\theta;\lambda).
\end{split}
\end{equation}
The formal solution to this equation is given by
\begin{equation}\label{DT}
\begin{split}
\theta(x) &= \lambda \left( 1 + x\int_0^1 dv v^2 f(vx,\theta(vx);\lambda)   \right) \\
& \equiv T_{\lambda}(\theta)(x),
\end{split}
\end{equation}
where $T_{\lambda}$ is a nonlinear map on a functional space. Equations (\ref{nulleq}) 
and (\ref{DT}) imply that the fixed point of $T_{\lambda}$ can be a solution to equation (\ref{nulleq}). 
Thus, if $T_{\lambda}$ has a fixed point on proper subset of $C^0$ on proper domain, we can prove the 
existence of the solution satisfying (\ref{nulleq}), which emanates from the central singularity,
and then it means that the singularity is naked. In the next section, using the fixed point theorem for a 
compact operator introduced soon, we examine the condition that central singularity is naked.

\section{Existence of the null geodesics emanating from the central singularity}

\subsection{Preparation}

In four dimensional case \cite{Christodoulou:1984mz}, the existence of null geodesics is shown by using 
the fixed-point theorem for contraction mapping \cite{banach1922operations}. In five dimensional case, however, we can not use the same 
method (see Appendix B). Thus, we have to innovate new one.

Firstly, we introduce a fixed point theorem which is suitable for the current issue \cite{Zeidler:1986}. 

\begin{theorem*}[Schauder fixed-point theorem  \cite{Schauder:1930}]
Let $D$ be a nonempty, closed, bounded, convex subset of a Banach space $X$, and suppose $T: D \to D$ is a compact operator. Then $T$ has a fixed point.
\end{theorem*}

In the above, compact operator is defined as follows.
\begin{definition*}[Compact Operator]
A operator $T$ is compact if and only if :\\
(1) $T$ is continuous. \\
(2) $T$ maps bounded set into relatively compact set.
\end{definition*}
And, ``relative compact" means that the closure is compact. Note that $T$ does not have to be a linear operator.

We also use Arzel\`a-Ascoli theorem so as to show the relatively compactness of the image of $T$. 

\begin{theorem*}[Arzel\`a-Ascoli Theorem]
Let G be a nonempty open set in $\mathbb{R}^n$. A subset M in $C^0(\bar{G}) $ is relatively compact if and only if the following two conditions hold.\\
(i)(uniform boundedness)
\begin{equation}
\sup_{f\in M}\: (\:\sup_{x\in \bar{G}} \:|f(x)|\:) < \infty.
\end{equation}
(ii)(equicontinuity)
For arbitrary $\epsilon > 0$, there exists $\delta>0$, which depends only on $\epsilon$, such that 
\begin{equation}
\sup_{f\in M} |f(x) - f(y)| < \epsilon
\end{equation}
for each $x$, $y\in \bar{G}$ satisfying $|x-y|<\delta$.
\end{theorem*}

We apply Schauder Fixed-Point Theorem to the non-linear operator $T_{\lambda}$ defined in equation (\ref{DT}) and ask 
if the equation for null line has a solution. First of all, we introduce a domain such that $T_{\lambda}$ maps its domain 
into itself. So let us define
\begin{equation}\label{DD}
D_{\lambda,b,c,d} \equiv \{ \theta | \theta \in C^0[0,d] , |\theta - \lambda| \leq bx^c\},
\end{equation}
where $b,c,d\in \mathbb{R}$ satisfy $b > 0$, $c > 0$, $l \geq d > 0$ and $\lambda + bx^c < \theta_S(x)$ for all $x$ in $[0,l]$. Last inequality 
can be always satisfied for sufficient small $d$ because $\theta_S$ is a continuous function and $\lambda < \frac{1}{4} = \theta_S(0)$ always holds. 
Here we introduce the uniform norm $\|\theta\| \equiv \sup_{x \in [0,d]}\theta(x)$ such that $D_{\lambda,b,c,d}$ becomes 
a subset of a Banach space $C^0[0,d]$.
Then we can show following lemma.

\begin{lemma}
For all $\lambda < \frac{9-\sqrt{33}}{16}$, there exist $c(\lambda) \in (0,1)$, which depend only on $\lambda$, and $\bar{d} \in (0,l]$, such that $T_{\lambda}:D_{\lambda,b,c,d} \to D_{\lambda,b,c,d}$ for all $c \in[c(\lambda),1)$ and $d \in (0,\bar{d}]$.
\end{lemma}
 
\begin{proof}
For $\theta \in D_{\lambda,b,c,d}$, we estimate the right-hand side of equation (\ref{nulleq}) as 
\begin{equation}\label{ineq0}
\begin{split}
 | \lambda f(x,\theta(x);\lambda) | &=  \left|\frac{   1 - a(x)(\theta(x) x^2 + 1)^2 - \frac{x}{2}(\theta(x) x^2+1)^2\frac{d}{dx}a(x)    }{x^2\sqrt{\gamma \left( 1-\frac{a(x)x^2}{\gamma} \right) \left( 1-a(x)(\theta(x) x^2+1)^2 \right) }} - \frac{2\lambda}{x} \right| \\
& \leq  \frac{ 2\lambda }{\sqrt{\gamma} \left( \frac{1}{2} - 2\lambda \right)^{\frac{3}{2}}}bx^{c-1} +O(1) + O(x^{2c-1}),
\end{split}
\end{equation}
where we used (\ref{lambda}) and (\ref{DD}) (the details are shown in Appendix A) and the terms $O(1)$ and $O(x^{2c-1})$ does not depend on $\theta$. Then we see
\begin{equation}\label{Testimate}
 |T_{\lambda}(\theta) - \lambda| 
 \leq  x\int_0^1 dv v^2 | \lambda f(vx,\theta(vx);\lambda) | 
= \frac{ 2\lambda }{\sqrt{\gamma} (c+2)\left( \frac{1}{2} - 2\lambda \right)^{\frac{3}{2}}}bx^c  + O(x^{2c}) + O(x).
\end{equation}
Hence, if
\begin{equation}
\label{IEc}
0 < c < 1
\end{equation}
and
\begin{equation}
\label{IEco}
\frac{ 2\lambda }{\sqrt{\gamma} (c+2)\left( \frac{1}{2} - 2\lambda \right)^{\frac{3}{2}}} < 1
\end{equation}
hold and $d$ is sufficiently small, $|T_{\lambda}(\theta) - \lambda| \leq bx^c$ holds for all 
$x$ in $[0,d]$. This means that there exists a positive number $\bar{d}$ such that $T_{\lambda}(\theta)$ is in $D_{\lambda,b,c,d}$ for arbitrary $d$ in $(0,\bar{d}]$, that is, 
$T_{\lambda}$ maps $D_{\lambda,b,c,d}$ into itself for such $d$. 

There exists $c$ such that the equations (\ref{IEc}) and (\ref{IEco}) hold if and only if 
\begin{equation}
\label{lambdacondi}
\lambda < \frac{9-\sqrt{33}}{16}
\end{equation}
holds. In this case, the equations (\ref{IEc}) and (\ref{IEco}) imply
\begin{equation}
\frac{8\lambda^2}{(1-2\lambda)(1-4\lambda)}-2 < c < 1
\end{equation}
and then we can take the parameter $c$ in this range so as to satisfy (\ref{IEc}) and (\ref{IEco}). For example, 
$c(\lambda)$ in the Lemma 1 is a number sligtly larger than $\frac{8\lambda^2}{(1-2\lambda)(1-4\lambda)}-2$. 
\end{proof}

From (\ref{lamdbapm}) and (\ref{lambdacondi}), for all $\lambda_-(\gamma)$, we can take some $c$ 
satisfying equations (\ref{IEc}) and (\ref{IEco}) because of $\lambda_-(\gamma) \leq \lambda_{M}< \frac{9-\sqrt{33}}{16}$. Thus, we 
can choose $c$ and $d$ such that $T_{\lambda}:D_{\lambda,b,c,d} \to D_{\lambda,b,c,d} $ for all 
$\lambda_-(\gamma)$.

Next, we evaluate $T_{\lambda}(\theta)$ and show that it is uniformly continuous on $D_{\lambda,b,c,d}$.

\begin{lemma}
$T_{\lambda}:D_{\lambda,b,c,d} \to D_{\lambda,b,c,d} $ is uniformly continuous.
\end{lemma}

\begin{proof}

At first, we evaluate the absolute value of difference of integrand in $T_\lambda(\theta)$ for different 
$\theta_1$ and $\theta_2$ in $D_{\lambda,b,c,d}$. For the convenience, let us define
\begin{equation}
g_i \equiv 1 - a(x)\Bigl(\theta_i(x) x^2 + 1\Bigr)^2,
\end{equation}
where the index $i$ takes $1$ or $2$. Since we have $\theta_i(x) < \theta_S(x)$ from the definition (\ref{DD}) and 
equation (\ref{locationofS}) shows us $1 - a(x)(\theta_S(x) x^2 + 1)^2 = 0$, $g_i$ is always positive. Then, we obtain
\begin{equation}\label{ineq1}
\begin{split}
&|\lambda f(x,\theta_1(x);\lambda) - \lambda f(x,\theta_2(x);\lambda)| \\
&= \left| 1 +  \frac{x\frac{d}{dx}a(x)}{2a(x)} \left( 1 + \frac{1}{\sqrt{g_1 g_2}}\right)\right|\frac{\left| \sqrt{g_1}-\sqrt{g_2}\right| }{x^2\sqrt{\gamma \left( 1-\frac{a(x)x^2}{\gamma} \right) }}  \\
& \leq \left\{ 1+ \left| \frac{x\frac{d}{dx}a(x)}{2a(x)} \right| \left( 1 + \frac{1}{\sqrt{g_1g_2}} \right) \right\}  \left( \frac{1}{\sqrt{g_1} + \sqrt{g_2}} \right) \frac{|g_1 - g_2|}{x^2\sqrt{\gamma \left( 1-\frac{a(x)x^2}{\gamma} \right) }} \\
& \leq \left\{ 1- \frac{x\frac{d}{dx}a(x)}{2a(x)}  \left( 1 + \frac{1}{1-a(x)\{(\lambda+bx^c)x^2 +1\}} \right) \right\}\\
&\;\;\;\;\;\;\;\;\;\;\;\;\;\;\;\;\;\;\;\;\;\;\;\;\;\;\frac{a(x) \{x^4(\lambda + bx^c) + x^2 \}}{\sqrt{1-a(x)\{(\lambda+bx^c)x^2+1\}^2}}\frac{|\theta_1(x)-\theta_2(x)|}{x^2\sqrt{\gamma \left( 1-\frac{a(x)x^2}{\gamma} \right) }}\\
&= \left(\frac{4\lambda}{1-4\lambda}x^{-1}  + B_1(x)x^{\delta - 1} \right) |\theta_1(x) - \theta_2(x)|,
\end{split}
\end{equation}
where $\delta = \min(1,c)$ and $B_1(x)\in C^0[0,d]$ is a positive function which dose not depend on $d$, $\theta_1$ and $\theta_2$, and we used equations (\ref{monoa}), 
(\ref{lambda}) and the fact
\begin{equation}
\begin{split}
|g_1- g_2| &= |a(x)| |(\theta_1(x)x^2 + 1)^2 - (\theta_2(x)x^2 + 1)^2| \\
&\leq 2 |a(x)| | x^4(\lambda + b x^c) + x^2 | | \theta_1(x) - \theta_2(x) |
\end{split}
\end{equation}
(See Appendix A for the details). Using (\ref{ineq1}), we obtain
\begin{eqnarray}\label{differenceofTT}
\|T_{\lambda}(\theta_1) - T_{\lambda}(\theta_2) \| &=& \sup_{0\leq x \leq d} \left| x \int_0^1 v^2\Bigl\{ \lambda f(vx,\theta_1(vx);\lambda) - \lambda f(vx,\theta_2(vx);\lambda) \Bigr\}dv \right| \nonumber\\
&\leq& \sup_{0\leq x \leq d} \frac{1}{x^2} \int_0^x y^2\left(\frac{4\lambda}{1-4\lambda}y^{-1}  + B_1(y)y^{\delta - 1} \right) |\theta_1(y) - \theta_2(y)| dy\nonumber\\
&\leq& \| \theta_1- \theta_2\| \sup_{0\leq x \leq d} \frac{1}{x^2} \int_0^x \left(\frac{4\lambda}{1-4\lambda}y  + B_1(y)y^{\delta + 1} \right) dy\nonumber\\
&=&\| \theta_1- \theta_2\| \left(\frac{2\lambda}{1-4\lambda}  + \sup_{0\leq x \leq d} \frac{1}{x^2} \int_0^x B_1(y)y^{\delta + 1} dy\right).
\end{eqnarray}
The integral of the right-hand side of inequality is finite because $B_1 \in C^0[0,d]$ and $\delta > 0$. 
Therefore, $T_{\lambda}:D_{\lambda,b,c,d} \to D_{\lambda,b,c,d} $ is uniformly continuous. 
\end{proof}

Moreover, for specific initial condition, $T_{\lambda}$ becomes a contraction mapping. In this case, as below, we can 
immediately show the existence of solution to equation (\ref{DT}).

\begin{theorem}
For all $\lambda < \frac{1}{6}$, there exist $d \in (0,l ] $($l$ corresponds to the surface of dust cloud) 
and a unique solution $\theta \in C^{\infty}(0,d]$ 
to the integral equation (\ref{DT}), which is continuous at $x = 0$ and satisfies $\theta(0) = \lambda$.
\end{theorem}

\begin{proof}
The last term of (\ref{differenceofTT}) is estimated as
\begin{equation}
\frac{1}{x^2} \int_0^x B_1(y)y^{\delta + 1} dy = O(x^\delta).
\end{equation}
Here, if
\begin{equation}\label{Ccontractive}
\frac{2\lambda}{1-4\lambda} < 1,
\end{equation}
we can choose sufficient small $d$ such that
\begin{equation}
\frac{2\lambda}{1-4\lambda}  + \sup_{0\leq x \leq d} \frac{1}{x^2} \int_0^x B_1(y)y^{\delta + 1} dy < 1,
\end{equation}
that is, $T_{\lambda}:D_{\lambda,b,c,d} \to D_{\lambda,b,c,d} $ is contractive. Thus, by the fixed-point theorem for 
contraction mapping \cite{banach1922operations}, $T_{\lambda}$ has a unique fixed point $\theta \in D_{\lambda,b,c,d}$. 
Note that the condition (\ref{Ccontractive}) is equivalent to $\lambda < \frac{1}{6}$. 
\end{proof}

Since the integrand of 
right-hand side of equation (\ref{DT}) is $C^{\infty}$ function in the region except for the singularity, the solution $\theta$ must be $C^{\infty}$ function in $(0,d]$.
On the other hand, the solution $\theta(x) \in C^0[0,d]$ is not differentiable at $x = 0$ in general. 
However, $\zeta(x) = x^2\theta(x)$ is in $C^1[0,d]$ because $\zeta(x) \in C^\infty (0,d]$ and $\frac{d}{dx}\zeta(x)$ is 
$O(x)$ from (\ref{nulleq_1}). This means that the differential equation of null line (\ref{nulleq_1}) has a solution 
$\zeta(x) \in C^1[0,d]$ which is a future directed outgoing null geodesic emanating from the central singularity for 
all $\lambda < \frac{1}{6}$.

For $\gamma > 24$, the equation (\ref{lambda}) tells us $\lambda _-(\gamma) < \frac{1}{6}$. Then, in this case, 
there exists a null line emanating from the central singularity, that is, it is a locally naked singularity at least.

\subsection{A proof of existence of the null geodesics}

In the case of $24 \leq \gamma \leq \gamma_{min}$, that is, $\frac{1}{6} \leq \lambda _-(\gamma) $ holds, 
we can not use Theorem 2 and show the existence of solution to equation (\ref{DT}). Then we have to develop 
another method. As we already mentioned, we can show the existence of solution to equation (\ref{DT}) by 
using of Schauder fixed-point theorem.

\begin{theorem}
For all $\lambda < \frac{9-\sqrt{33}}{16}$, there exist $d \in (0,l]$ and a solution $\theta \in C^{\infty}(0,d]$ 
to the integral equation (\ref{DT}), which is continuous at $x = 0$ and satisfies $\theta(0) = \lambda$. 
\end{theorem}

\begin{proof}
From Schauder fixed-point theorem, if $D_{\lambda,b,c,d}$ is a nonempty, closed, bounded, convex subset of a Banach space, 
and $T_\lambda$ maps $D_{\lambda,b,c,d}$ into itself and is compact operator, then $T_\lambda$ has a fixed point. 
We already showed in Lemma 1 that  we can take certain numbers $d\in(0,l]$ and $c \in (0,1)$ such that $T_\lambda$ maps $D_{\lambda,b,c,d}$ 
into itself for all $\lambda < \frac{9-\sqrt{33}}{16}$. Moreover, Lemma 2 tells us that $T_\lambda$ is continuous. By the 
definition (\ref{DD}), $D_{\lambda,b,c,d}$ is obviously nonempty, closed, bounded subset of a Banach space $C^0[0,d]$. Furthermore, for all 
$\theta_1, \theta_2 \in D_{\lambda,b,c,d}$ and $0 < \kappa < 1$,
\begin{equation}
\lambda - bx^c \leq \kappa\theta_1(x) + (1-\kappa)\theta_2(x) \leq \lambda + bx^c.
\end{equation}
This means that $D_{\lambda,b,c,d}$ is convex. Then all we have to show is that $T_\lambda(D_{\lambda,b,c,d})$ is relatively compact set. 
By virtue of Arzel\`a-Ascoli theorem, the remaining task is to show uniform boundedness and equicontinuity of $T_\lambda(D_{\lambda,b,c,d})$. 
Uniform boundedness results from boundedness of $D_{\lambda,b,c,d}$ as follows.
\begin{equation}
\begin{split}
\sup_{T_\lambda(\theta) \in T_\lambda(D_{\lambda,b,c,d})}\: (\:\sup_{x\in [0,d]} \:|T_\lambda(\theta)(x)|\:) &\leq \sup_{T_\lambda(\theta) \in T_\lambda(D_{\lambda,b,c,d})}\: (\:\sup_{x\in [0,d]} \: \lambda + bx^c\:)\\
&=\sup_{T_\lambda(\theta) \in T_\lambda(D_{\lambda,b,c,d})} (\lambda + bd^c)\\
&=\lambda + bd^c.
\end{split}
\end{equation}

Next, to show equicontinuity, we evaluate
\begin{equation}
\begin{split}
|T_\lambda(\theta)(x)-T_\lambda(\theta)(y)| &= \left| \frac{1}{x^2} \int^x_0 s^2\lambda f(s,\theta(s);\lambda) ds - \frac{1}{y^2} \int^y_0 s^2\lambda f(s,\theta(s);\lambda) ds\right|\\
&=\left| \left(\frac{1}{x^2} -\frac{1}{y^2} \right)\int^x_0 s^2\lambda f(s,\theta(s);\lambda) ds - \frac{1}{y^2} \int^y_x s^2\lambda f(s,\theta(s);\lambda) ds\right|.
\end{split}
\end{equation}
Here note that we can assume $x<y$ without loss of generality. Then, using equation (\ref{ineq0}),
\begin{equation}\label{Tineq}
\begin{split}
|T_\lambda(\theta)(x)-T_\lambda(\theta)(y)| &\leq \left(\frac{1}{x^2} -\frac{1}{y^2} \right)\int^x_0 s^2|\lambda f(s,\theta(s);\lambda)| ds +\frac{1}{y^2} \int^y_x s^2|\lambda f(s,\theta(s);\lambda)| ds\\
%&\leq \left(\frac{1}{x^2} -\frac{1}{y^2} \right)\int^x_0 s^2 \left\{\frac{ 2\lambda }{\sqrt{\gamma} \left( \frac{1}{2} - 2\lambda \right)^{\frac{3}{2}}}bs^{c-1}  + O(1) + O(s^{2c-1}) \right\}ds\\
%&+\frac{1}{y^2}\int^y_x s^2 \left\{\frac{ 2\lambda }{\sqrt{\gamma} \left( \frac{1}{2} - 2\lambda \right)^{\frac{3}{2}}}bs^{c-1}  + O(1) + O(s^{2c-1}\right\}ds\\
%&= \left(\frac{1}{x^2} -\frac{1}{y^2} \right)h(x)x^{2+c} + \frac{1}{y^2}(h(y)y^{2+c}-h(x)x^{2+c})\\
%&\leq h(x)|x^c-y^c|+y^c|h(y)-h(x)|+\frac{2h(x)}{y^2} |y^{2+c}-x^{2+c}|\\
%&< h(x)|x^c-y^c|+y^c|h(y)-h(x)|+2h(x)y^c \left| \frac{y^2-x^2}{y^2} \right|\\
&\leq h(x)|x^c-y^c|+y^c|h(x)-h(y)|+2^{n+1}h(x)y^{c-\frac{1}{2^{n-1}}} \left| x^{\frac{1}{2^{n-1}}}-y^{\frac{1}{2^{n-1}}} \right|,
\end{split}
\end{equation}
where $h(x)$ is a $C^0$ positive function in the range $[0,d]$, which dose not depend on $\theta$, and $n$ is an arbitrary natural number (the details of the calculation is in Appendix A). 
So we take $n$ such that 
\begin{equation}
c-\frac{1}{2^{n-1}} > 0.
\end{equation}
Then, since $h(x)$, $y^c$ and $2^{n+1}h(x)y^{c-\frac{1}{2^{n-1}}}$ are continuous functions for $(x,y)$ in $[0,d]\times[0,d]$, 
there exist real numbers $\Delta_1$, $\Delta_2$ and $\Delta_3$ such that
\begin{equation}
|T_\lambda(\theta)(x)-T_\lambda(\theta)(y)|  < \Delta_1 |x^c-y^c|+\Delta_2 |h(x)-h(y)| + \Delta_3\left| x^{\frac{1}{2^{n-1}}}-y^{\frac{1}{2^{n-1}}} \right|.
\end{equation}
Since all continuous functions on compact set is uniformly continuous, for all $\epsilon > 0$, there exist $\delta_1 >0$, $\delta_2>0$ and 
$\delta_3>0$ which are independent of $x$ and $y$, and $\Delta_1|x^c-y^c|<\frac{\epsilon}{3}$, $ \Delta_2|h(x)-h(y)|<\frac{\epsilon}{3}$, 
$\Delta_3\left| x^{\frac{1}{2^{n-1}}}-y^{\frac{1}{2^{n-1}}} \right|<\frac{\epsilon}{3}$ hold whenever $|x-y| < \delta_1$, 
$|x-y| < \delta_2$, $|x-y| < \delta_3$ respectively. Thus, $ |T_\lambda(\theta)(x)-T_\lambda(\theta)(y)|  < \epsilon$ holds whenever 
$|x-y| < \delta \equiv \min\{  \delta_1, \delta_2, \delta_3\}$ holds. This means that $T_\lambda(D_{\lambda,b,c,d})$ has equicontinuity. 
Therefore, $T_\lambda(D_{\lambda,b,c,d})$ is relatively compact set. Since any closed set included in a compact set is also compact, $T$ maps any bounded set into relatively compact set. Thus $T$ is compact operator.
\end{proof}

In the same way as the discussion below theorem 2, theorem 3 means that the differential equation for null line (\ref{nulleq_1}) has a solution $\zeta \in C^1[0,d]$ 
which is a future directed outgoing null geodesic emanating from the central singularity for all $\lambda < \frac{9-\sqrt{33}}{16}$. Thus, in the following, if the function is a solution to equation (\ref{nulleq}) converge to finite value as $x \to 0$, we consider the function as the solution to equation (\ref{nulleq}), which is also defined at $x=0$.
Note that the solution found in Theorem 2 is unique, but it is not necessary that the solution found in Thoerem 3 is unique.

If $\lambda$ satisfying equation (\ref{lambda}) exists, $\lambda_-(\gamma)$ always satisfies $\lambda_-(\gamma) \leq\lambda_{M} < \frac{9-\sqrt{33}}{16}$. 
From theorem 3, this fact means that $T_{\lambda_-(\gamma)}$ has a fixed point which converge to $\lambda_-(\gamma)$ as $x \to 0$, that is, there exists 
a null geodesic emanating from the central singularity. Thus, the central singularity must be locally naked at least if $\lambda$ satisfying 
equation (\ref{lambda}) exists.

\section{Spacetime structure around the singularity}

In this section, we show the existence of the earliest null geodesic $\theta_{n_0}$ emanating from the central singularity for 
all $\gamma \geq \gamma_{min}={\sqrt {11+5\sqrt {5}}}$. Since such null geodesic determines the causal structure around the naked singularity and the 
global nakedness of the singularity, $\theta_{n_0}$ plays an important role in our analysis. On the other hand, for 
$\gamma < \gamma_{min}$, we also show that there is no causal geodesic emanating from the central singularity. 

In the case of four dimension, we can show that $g_{rr}$ of LTB spacetime is strictly monotonically decreasing function with respect to 
$t$ near the singularity. Using this nature, we can immediately specify $\theta_{n_0}$ (see Ref. \cite{Christodoulou:1984mz}). 
By contrast, in the case of five dimension, $g_{rr}$ becomes a monotonically increasing function with respect to $t$ near the singularity. Therefore, 
we need to develop an another method to specify $\theta_{n_0}$, which is general to some extent.

At first, we show that any future directed null geodesic along the outer radial direction can not emanate from the central singularity located at $\theta <\lambda_-(\gamma)$.

\begin{lemma}
If $\lambda$ satisfying equation (\ref{lambda}) exists and a future directed null geodesic along the outer radial direction, 
$\theta(x)$, converges as $x \to 0$, then $\theta(0) \geq \lambda_-(\gamma)$ holds. On the other hand, if $\lambda$ satisfying equation 
(\ref{lambda}) does not exist, any future directed outgoing causal line, $\theta_c(x)$, satisfies $\theta_c(x) \to -\infty$ as $x \to 0$.
\end{lemma}

\begin{proof}
Let us assume that a solution to equation (\ref{nulleq}), $\theta(x)$, is a $C^0$ function in the range $(0,d]$ and is a 
bounded function. Bearing equation (\ref{limf}) in mind, equation (\ref{nulleq}) for $\theta(x)$ can be written as
\begin{equation}\label{Lemma3_1}
\begin{split}
\frac{d\theta(x)}{dx} &=  \frac{   1 - a(x)(\theta(x) x^2 + 1)^2 - \frac{x}{2}(\theta(x) x^2+1)^2\frac{d}{dx}a(x)    }{\sqrt{\gamma \left( 1-\frac{a(x)x^2}{\gamma} \right) \left( 1-a(x)(\theta(x) x^2+1)^2 \right) }} - \frac{2\theta(x)}{x} \\
&= \frac{ 1- 2\theta(x) + O(x) }{ x\sqrt {\gamma \left( \frac{1}{2} - 2\theta(x) +O(x) \right)}   } - \frac{2\theta(x)}{x}.
\end{split}
\end{equation}
Moreover, if the solution $\theta(x)$ converges as $x \to 0$, then, for all $\epsilon > 0$, there exists $d_0 > 0$ such that
\begin{equation}\label{Lemma3_2}
\begin{split}
\frac{ 1- 2\theta(x) + O(x) }{ x\sqrt {\gamma \left( \frac{1}{2} - 2\theta(x) +O(x) \right)}   } - \frac{2\theta(x)}{x} &> \frac{ 1- 2\theta(0)}{ x\sqrt {\gamma \left( \frac{1}{2} - 2\theta(0) \right)}   } - \frac{2\theta(0)}{x} - \frac{\epsilon}{x}\\
&\equiv \frac{g(\theta(0);\gamma)}{x} - \frac{\epsilon}{x}
\end{split}
\end{equation}
for arbitrary $x \in (0, d_0]$, where
\begin{equation}\label{Dg}
g(\theta;\gamma) \equiv \frac{1-2\theta}{\sqrt{\gamma(\frac{1}{2}-2\theta)}}-2\theta.
\end{equation}
In the case that $\lambda$ satisfying equation (\ref{lambda}) exists, $g(\theta;\gamma)$ satisfies
\begin{numcases}
{}
g(\lambda_\pm(\gamma);\gamma) = 0 & \nonumber \\
g(\theta;\gamma) < 0 & $( \lambda_-(\gamma) < \theta < \lambda_+(\gamma) )$ \nonumber \\
g(\theta;\gamma) > 0  &  $( \theta < \lambda_-(\gamma), \lambda_+(\gamma) < \theta < \frac{1}{4})$ .
\end{numcases}
In the case that $\lambda_\pm(\gamma)$ does not exist, for all $\theta < \frac{1}{4}$,
\begin{equation}\label{Lemma3_3}
g(\theta;\gamma) > 0
\end{equation}
and $\inf_{\theta <\frac{1}{4}} g(\theta;\gamma) > 0$ hold.

Now we assume that  $\lambda$ satisfying equation (\ref{lambda}) exists and a solution to equation (\ref{nulleq}), 
$\theta(x)$, satisfies $\theta(0) < \lambda_-(\gamma)$. In this case, we can choose $\epsilon$ such that
\begin{equation}\label{Lemma3_4}
g(\theta(0);\gamma)- \epsilon = \kappa_0^2,
\end{equation}
where $\kappa_0$ is a real number. Then, if the solution $\theta(x)$ converges as $x \to 0$, 
from equations (\ref{Lemma3_1}), (\ref{Lemma3_2}) and (\ref{Lemma3_4}), there exists $d_0 > 0$ such that
\begin{equation}
\frac{d\theta(x)}{dx} > \frac{\kappa_0^2}{x}
\end{equation}
for arbitrary $x \in (0, d_0]$. Integrating this equation, we obtain
\begin{equation}
\theta(x) - \theta(0) > \lim_{x_0 \to 0}\kappa_0^2 \Big[ \ln{x} \Big]^x_{x_0}.
\end{equation}
In the above, the right-hand side diverges. It contradicts the assumption that $\theta(x)$ is bounded. 
Thus, $\theta(0) \geq \lambda_-(\gamma)$, that is, any future directed null geodesic along the outer radial 
direction does not converge to the central singularity that satisfies $\theta < \lambda_-(\gamma)$. 

On the other hand, in the case that $\lambda$ satisfying equation (\ref{lambda}) does not exist, 
equations (\ref{Lemma3_1}) and (\ref{Lemma3_3}) hold for any bounded solution $\theta(x)$ satisfying $\theta (x)< \frac{1}{4}$. In addition, 
since any future directed null geodesics can not emerge from the apparent horizon determined by 
equation (\ref{locationofAH}), any solution that can extend $x \to 0$ must not enter the region 
$\theta > \theta_{AH}(x;\gamma)$ as $x \to 0$. Then, there exist $d_1 > 0$ and $\mu > 0$ such that 
the solution satisfies $\theta(x) < \frac{1}{4} - \mu$ for all $x \in [0,d_1]$ because $\theta_{AH} (x;\gamma)$ is continuous and $\theta_{AH} (0;\gamma)= \frac{1}{4} - \frac{1}{2\gamma}$ holds. In this region, $g(\theta;\lambda)$ 
 is defined. Then, for all $\epsilon > 0$, there exists $d_2$ satisfying $d_1 \geq d_2 > 0$ such that 
\begin{equation}
\begin{split}
\frac{d\theta(x)}{dx}=\frac{ 1- 2\theta(x) + O(x) }{ x\sqrt {\gamma \left( \frac{1}{2} - 2\theta(x) +O(x) \right)}}
-\frac{2\theta(x)}{x} &> \frac{g(\theta(x);\gamma)}{x} - \frac{\epsilon}{x}\\
& \geq \frac{\inf_{\theta <\frac{1}{4}} g(\theta;\gamma)}{x} - \frac{\epsilon}{x}
\end{split}
\end{equation}
for arbitrary $x \in (0, d_2]$. Note that we do not assume the solution $\theta(x)$ converges as $x \to 0$ here. Since $\lambda$ satisfying equation (\ref{lambda}) does not exist, 
$\inf_{\theta <\frac{1}{4}} g(\theta;\gamma) > 0$. Then we can choose $\epsilon$ such as
\begin{equation}
\inf_{\theta <\frac{1}{4}} g(\theta;\gamma)- \epsilon = \kappa_1^2,
\end{equation}
where $\kappa_1$ is a real number. Thus, in the same way as the case that $\lambda$ satisfying equation (\ref{lambda}) exists, 
we can show that any solution $\theta(x)$ can not be bounded. Therefore, any future directed null geodesic along 
the outer radial direction must diverge to $-\infty$ as $x \to 0$. Since the future directed null geodesic 
along the outer radial direction is obviously the earliest line that emerged from $x = 0$ at arbitrary time, 
any future directed outgoing causal line must diverge to $-\infty$ at $x = 0$.

\end{proof}

Thus, for the case that $\lambda$ satisfying equation (\ref{lambda}) {\it dose not exist}, there is no causal line 
which emanates from the central singularity, that is, strong cosmic censorship holds. In contrast, for the case that 
$\lambda$ satisfying equation (\ref{lambda}) {\it exists}, we have just discussed converged null geodesics and have 
not yet shown anything about another causal lines that emanate from the central singularity. Then, to address this point, 
we will first present the following Lemma 4 and 5.

\begin{lemma}
Let $\theta(x)$ be a future directed outgoing null geodesic along radial direction and oscillates as $x \to 0$. 
If $\lambda$ satisfying equation (\ref{lambda}) exists, then $\varliminf_{x \to 0} \theta(x) \geq \lambda_-(\gamma)$ holds.
\end{lemma}

\begin{proof}
We assume that $\lambda$ satisfying equation (\ref{lambda}) exists and a future directed outgoing null geodesic 
along radial direction, $\theta(x)$, is a $C^0$ function in the range $(0,d]$ and is oscillated as $x \to 0$. 
If $\varliminf_{x \to 0} \theta(x) < \lambda_-(\gamma)$ holds, we can choose $\mu_1,\mu_2 \in \mathbb{R}$ such that
\begin{equation}
\varliminf_{x \to 0} \theta(x) < \mu_1 <\mu_2 < \min\{\varlimsup_{x \to 0} \theta(x),\lambda_-(\gamma)\}.
\end{equation}
Here let us define
\begin{equation}\label{tildetheta}
\tilde{\theta}(x) = \begin{cases}
\mu_1 & (\theta(x) \leq \mu_1)\\
\theta(x) & (\mu_1 < \theta(x) < \mu_2)\\
\mu_2 & (\mu_2 \leq \theta(x) ).
\end{cases}
\end{equation}
It is easy to show that $\tilde{\theta}$ is a $C^0$ bounded function in the range $(0,d]$ and satisfies 
$\varliminf_{x \to 0} \tilde{\theta}(x) = \mu_1$, $\varlimsup_{x \to 0} \tilde{\theta}(x) = \mu_2$ and
\begin{equation}\label{Lemma4_1}
\frac{d\tilde{\theta}}{dx}(x) = \begin{cases}
0 & (\theta(x) < \mu_1)\\
 \frac{   1 - a(x)(\theta(x) x^2 + 1)^2 - \frac{x}{2}(\theta(x) x^2+1)^2\frac{d}{dx}a(x)    }{\sqrt{\gamma \left( 1-\frac{a(x)x^2}{\gamma} \right) \left( 1-a(x)(\theta(x) x^2+1)^2 \right) }} - \frac{2\theta(x)}{x}  & (\mu_1 < \theta(x) <\mu_2)\\
0 & (\mu_2 < \theta(x) ).
\end{cases}
\end{equation}

For all $x$ which satisfy $\mu_1 < \theta(x) <\mu_2 < \lambda_-(\gamma)< \frac{1}{4}$, the function $g(\theta(x);\gamma)$ 
exists as a real function. Then, as with the proof of Lemma 3, for all $\epsilon > 0$, there exists a $d_0 > 0$  such that 
\begin{equation}\label{Lemma4_2}
\begin{split}
 \frac{   1 - a(x)(\theta(x) x^2 + 1)^2 - \frac{x}{2}(\theta(x) x^2+1)^2\frac{d}{dx}a(x)    }{\sqrt{\gamma \left( 1-\frac{a(x)x^2}{\gamma} \right) \left( 1-a(x)(\theta(x) x^2+1)^2 \right) }} - \frac{2\theta(x)}{x} &> \frac{g(\theta(x);\gamma)}{x} - \frac{\epsilon}{x}\\
 &\geq \frac{\inf_{\mu_1<\theta <\mu_2} g(\theta;\gamma)}{x} - \frac{\epsilon}{x}
\end{split}
\end{equation}
for arbitrary $x \in (0, d_0]$, which satisfy $\mu_1 < \theta(x) <\mu_2$. Here, $\inf_{\mu_1<\theta <\mu_2} g(\theta;\gamma) > 0$ 
because of $\mu_2 < \lambda_-(\gamma)$. Then we can choose $\epsilon$ such that
\begin{equation}\label{Lemma4_3}
\inf_{\mu_1<\theta <\mu_2} g(\theta;\gamma)- \epsilon = \kappa^2,
\end{equation}
where $\kappa$ is a real number. Equations (\ref{Lemma4_1}), (\ref{Lemma4_2}) and (\ref{Lemma4_3}) imply
\begin{equation}\label{Lemma4_4}
\frac{d\tilde{\theta}}{dx}(x) \geq \begin{cases}
0 & (\theta(x) < \mu_1)\\
\frac{\kappa^2}{x} & (\mu_1 < \theta(x) <\mu_2)\\
0 & (\mu_2 < \theta(x) )
\end{cases}
\end{equation}
for arbitrary $x \in (0, d_0]$. Thus, from equation (\ref{tildetheta}) and (\ref{Lemma4_4}), $\tilde{\theta}(x)$ is a monotonically increasing function. 
It contradicts the assumption that $\theta(x)$ oscillates as $x \to 0$. Therefore, the solution $\theta(x)$ 
must satisfy $\varliminf_{x \to 0} \theta(x) \geq \lambda_-(\gamma)$.

\end{proof}

Thus, from Lemma 3 and 4, we conclude that any future directed outgoing null geodesic along radial direction, $\theta(x)$, 
satisfies (i) $\varliminf_{x \to 0} \theta(x) \geq \lambda_-(\gamma)$, or (ii)  $\theta(x)$ diverges to $-\infty$ as $x \to 0$ 
if $\lambda$ satisfying equation (\ref{lambda}) exists. However, 
unlike the case that $\lambda$ satisfying equation (\ref{lambda}) does not exist, it immediately dose not mean that any 
future directed outgoing causal line becomes to satisfy $\theta \geq \lambda_-(\gamma)$ or $\theta \to -\infty$ as $x \to 0$ 
because a null line which converges to the central singularity exists in this case. Hence, we have to carefully examine the 
geodesics in this case.

For $d > 0$, let us define $G_{\lambda_-(\gamma),d} \subset C^0[0,d]$ as the set of the solutions to equation (\ref{nulleq}) for 
$\gamma$, which converges to $\lambda_-(\gamma)$ as $x \to 0$ and does not enter singularity at a point in $(0,d]$. Then, we can show the following 
Lemma for $G_{\lambda_-(\gamma),d}$.

\begin{lemma}
If $\lambda$ satisfying equation (\ref{lambda}) exists, then there exist a solution to equation (\ref{nulleq}), $\theta_{n_0}(x;\gamma)$, 
such that $\theta_{n_0}(x;\gamma) \leq \theta(x)$ for all $\theta \in G_{\lambda_-(\gamma),d} $ and $x \in [0,d]$, which $d$ is an arbitrary positive number. 
Moreover, if the curve $\theta = \theta_{n_0}(x;\gamma)$ enters the singularity at $x = d_1 > 0$, $G_{\lambda_-(\gamma),d}$ 
is empty for all $d$ satisfying $d\geq d_1$.
\end{lemma}

\begin{proof}
We suppose that $\lambda$ satisfying equation (\ref{lambda}) exists. Let us define
\begin{equation}
G_{\lambda_-(\gamma),d}(x) \equiv \{\theta(x) | \theta \in G_{\lambda_-(\gamma),d}\},
\end{equation}
where $x \leq d$. By Theorem 3, there exists $d_0>0$ such that $G_{\lambda_-(\gamma),d_0}(d_0)$ is not empty. 
If a solution to equation (\ref{nulleq}) exists in $[0,d_0]$, then, for arbitrary $d \in (0,d_0]$, this solution also exists 
in $[0,d]$. Then,  $G_{\lambda_-(\gamma),d}(x)$ is not empty for all $d \in (0,d_0]$ and $x \in (0,d]$. 

Now we suppose that $G_{\lambda_-(\gamma),d_0}(x)$ is not empty for all $x \in (0,d_0]$. Since the right-hand side of 
equation (\ref{nulleq}) satisfies the Lifshitz condition on arbitrary closed set, which does not contain the singularity, all solutions to equation (\ref{nulleq}) do not intersect each other and can extend arbitrarily in any open set that does not contain the singularity. This fact means that the ordering 
of the solution orbits with respect to the coordinate $\theta$ is conserved. 

Let us define $\theta(x;x_0,\theta_0)$ as the solution to equation 
(\ref{nulleq}) such that passes through $(x_0,\theta_0)$. Then, for $ d_0\geq x_0>0$, Lemma 3, 4 and the above discussion tell 
us that $\theta(x;x_0,\inf G_{\lambda_-(\gamma),d_0}(x_0))$ must converge to $\lambda_-(\gamma)$ or diverge to $-\infty$ as $x \to 0$. 
Let us assume that $\theta(x;x_0,\inf G_{\lambda_-(\gamma),d_0}(x_0))$ diverges to $-\infty$ as $x \to 0$. 
Then, there exists $0 < x_1 < x_0$ such that $\theta(x_1;x_0,\inf G_{\lambda_-(\gamma),d_0}(x_0)) < 0$ holds. On such $x_1$, 
there exists the future directed solution to equation (\ref{nulleq}), $\theta_1(x)$, such that 
$\theta(x_1;x_0,\inf G_{\lambda_-(\gamma),d_0}(x_0))< \theta_1(x_1) < 0$. Since the region of $\theta < 0$ is not a future 
of the slice $\theta = 0$, Lemma 3 and 4 imply that $\theta_1(x)$ must diverge to $-\infty$ as $x \to 0$, that is, $\theta_1(x) \notin G_{\lambda_-(\gamma),d_0}$. 
Since the ordering of the solution orbits with respect to the coordinate $\theta$ is conserved, $\theta(x;x_0,\inf G_{\lambda_-(\gamma),d_0}(x_0))< \theta_1(x)$ holds for arbitrary point in the domain of $\theta_1$.
On the other hand, from the definition of $\theta(x;x_0,\inf G_{\lambda_-(\gamma),d_0}(x_0))$, for arbitrary $\epsilon>0$, there exits $\theta_{\epsilon}(x) \in G_{\lambda_-(\gamma),d_0}$ such that $0< \theta_{\epsilon}(x_0) - \theta(x_0;x_0,\inf G_{\lambda_-(\gamma),d_0}(x_0)) < \epsilon$ holds. $\theta_{\epsilon}(x)$ also satisfies $\theta_{\epsilon}(x) > \theta(x;x_0,\inf G_{\lambda_-(\gamma),d_0}(x_0))$ for arbitrary $x \in (0,x_0]$. 

Let us assume that $\theta_1(x)$ enters the singularity at $x = x_S(\theta_1)$ in $(x_1,x_0]$. 
Since $\theta_{\epsilon}(x) \in G_{\lambda_-(\gamma),d_0}$, $\theta_{\epsilon}(x)$ satisfies 
$0 < \theta_{\epsilon}(x) < \theta_S(x)$ for all $x \in (0,x_0]$. Then $\lim_{x \to x_S(\theta_1)}\theta_1(x) = \theta_S(x_S(\theta_1))> \theta_{\epsilon}(x_S(\theta_1))$  and $\theta_{\epsilon}(x_1) > 0 >\theta_1(x_1)$ hold, that is, $\theta_1(x)$ and $\theta_{\epsilon}(x)$ intersect at a point in $(x_1,x_S(\theta_1))$. It contradicts the fact that the right-hand side of equation (\ref{nulleq}) satisfies the Lifshitz condition. 
Then $\theta_1(x)$ does not enter the singularity in the range $(x_1,x_0]$. In this case, 
there exists $\epsilon'>0$ such that $\theta(x_0;x_0,\inf G_{\lambda_-(\gamma),d_0}(x_0)) < \theta_{\epsilon'}(x_0) < \theta_1(x_0)$ holds. 
Since $\theta_{\epsilon'}(x_1) > 0 >\theta_1(x_1)$ holds, it means that $\theta_1(x)$ and $\theta_{\epsilon'}(x)$ 
intersect at a point in $(x_1,x_0)$ and it leads to a contradict. 
Thus, $\theta(x;x_0,\inf G_{\lambda_-(\gamma),d_0}(x_0))$ converges to $\lambda_-(\gamma)$ as $x \to 0$, that is, 
$\theta(x;x_0,\inf G_{\lambda_-(\gamma),d_0}(x_0)) \in G_{\lambda_-(\gamma),d_0}$. 

Since the ordering 
of the solution orbits with respect to the coordinate $\theta$ is conserved and the solutions can extend arbitrarily in any open set that does not contain the singular points, for arbitrary $d$, 
$\theta(x;x_0,\inf G_{\lambda_-(\gamma),d_0}(x_0))$ must satisfy 
$\theta_G(x) \geq \theta(x;x_0,\inf G_{\lambda_-(\gamma),d_0}(x_0))$ for arbitrary $\theta_G \in G_{\lambda_-(\gamma),d}$ 
and $x \in [0,d]$. This means that $\theta(x;x_0,\inf G_{\lambda_-(\gamma),d_0}(x_0))$ must be $\theta_{n_0}(x;\gamma)$. 
Moreover, if the curve $\theta = \theta_{n_0}(x;\gamma)$ enters the singularity line $\theta = \theta_S(x)$ at some $x = d_1 > 0$, then all another solution lines, 
which emanate from singularity, are surrounded by the curve $\theta = \theta_{n_0}(x;\gamma)$ and $\theta = \theta_S(x)$. Since $\theta_{n_0}(x;\gamma)$ and 
any other solution line do not intersect each other except for the singularity, all solutions must intersect with singularity at a point in $(0,d_1]$. Then $G_{\lambda_-(\gamma),d}$ 
is empty for $d\geq d_1$.

\end{proof}

Lemma 3, 4 and 5 imply the following theorem. 

\begin{theorem}
(i) If $\lambda$ satisfying equation (\ref{lambda}) exists, then $\theta_{n_0}(x;\gamma)$ defined in Lemma 5 exists 
and is the earliest of all future directed causal line emanating from the central singularity. $\theta_{n_0}(x;\gamma)$ 
converges to $\lambda_-(\gamma)$ as $x \to 0$.

(ii) If $\lambda$ satisfying equation (\ref{lambda}) does not exist, then the strong cosmic censorship holds.

Furthermore, (i) and (ii) mean that $\lambda$ satisfying equation (\ref{lambda}) exists if and only if the central singularity is naked.
\end{theorem}

\begin{proof}
We already showed (ii) below Lemma 3. Then we will focus on (i). We suppose that a future directed causal line 
$\theta_c(x)$ satisfies $\theta_c(x_0) < \theta_{n_0}(x_0,\gamma)$ at a point $x = x_0 > 0$. From Lemma 3, 4 and 5, 
$\theta(x;x_0,\theta_c(x_0))$ must diverge to $-\infty$ as $x \to 0$. Since $\theta(x;x_0,\theta_c(x_0))$ corresponds 
to a future directed outgoing null geodesic along radial direction, $\theta_c(x)$ also diverges to $-\infty$ as $x \to 0$, 
that is, the line with $\theta_c(x)$ must emanate from the regular center. Then there is no future directed causal line which 
emanates from the central singularity before $ \theta_{n_0}(x;\gamma)$.
\end{proof}

From Theorem 4, if $\theta_{n_0}(x;\gamma)$ can extend to $x = l$ and $\theta_{n_0}(l;\gamma) < \theta_{AH}(l;\gamma)$ holds, 
the central singularity must be globally naked. Using this fact, in the next section, we consider the global 
structure of this spacetime.

\section{Global spacetime structure and the globally naked singularity}

In this section, we consider global property of singularity. We will see the dependence of the nakedness of 
the central singularity on the initial density distribution characterized by $\gamma$ and $a(x)$
(See equation (\ref{DPrho}) for the definitions). The discussion in this section is similar to
the four dimensional case \cite{Christodoulou:1984mz}.

\begin{lemma}
For any initial density distribution parameterized as (\ref{DPrho}), there exists $\gamma_0$ 
such that the solution $\theta_{n_0}(x;\gamma)$ can extend to $x = l$(corresponding to the surface of dust cloud) 
and $\theta_{n_0}(l;\gamma) < \theta_{AH}(l;\gamma)$ 
holds for all $\gamma \in [\gamma_0, \infty]$. $\theta_{n_0}$ is defined in Lemma 5. 
\end{lemma}

\begin{proof}
Since the outer region of the $x = l$ surface is the Schwarzschild spacetime and the event horizon is identical 
to the apparent horizon in the Schwarzschild spacetime, $\theta_{n_0}(l;\gamma) < \theta_{AH}(l;\gamma)$ means 
that the null line corresponding to $\theta_{n_0}(x;\gamma)$ arrives at outer region of the event horizon of 
the Schwarzschild spacetime, that is, the null line $\theta_{n_0}(x;\gamma)$ will attain the future null infinity and 
then the central singularity is globally naked.

To prove this lemma, by virtue of Theorem 3, it is enough to show that, for sufficient large $\gamma$, 
there exist $b$ and $c$ such that (i) $\lambda_-(\gamma) + bx^c < \theta_S(x)$ holds for all $x$ in $[0,l]$, 
(ii) $|T_{\lambda_-(\gamma)}(\theta)(x) - \lambda_-(\gamma)| \leq bx^c$ holds for all $x$ in $[0,l]$ and all $\theta$ in $C^0[0,l]$, 
which satisfy $|\theta(x) - \lambda_-(\gamma)| \leq bx^c$, (iii) $\theta_{n_0}(l;\gamma) < \theta_{AH}(l;\gamma)$ holds.

The conditions (\ref{Cofa1}), (\ref{Cofa3}) and (\ref{monoa}) imply $a(x) < 1$ for all $x$ in $(0,l]$. 
Since $\theta_S$ is continuous in the range $[0,l]$, we see
\begin{equation}
\min_{0 \leq x \leq l} \theta_S(x) = \min_{0 \leq x \leq l} \frac{1}{x^2} \left(\frac{1}{\sqrt{a(x)}} - 1\right) \equiv
\theta_{S,min} > 0. \label{defthetasmin}
\end{equation}
Since $\lambda_-(\gamma)$ is monotonically decreasing function such that it satisfies $\lim_{\gamma \to \infty} \lambda_-(\gamma) = 0$, 
there exists $\gamma_1 $ such that $\lambda_-(\gamma) < \theta_{S,min}$ holds for all $\gamma \in [\gamma_1, \infty)$. 
Then we can take a $b(\gamma_1) $ such that  
$\lambda_-(\gamma_1) + b(\gamma_1)l^c < \theta_{S,min}$ holds. For such $b(\gamma_1)$,
\begin{equation}
\lambda_-(\gamma) + b(\gamma_1)x^c \leq \lambda_-(\gamma_1) + b(\gamma_1)l^c <  \theta_{S,min} \leq \theta_S(x) \label{bgamma1}
\end{equation}
holds for all $\gamma$ in $[\gamma_1,\infty)$. This means that the condition (i) for all $\gamma$ in $[\gamma_1,\infty)$ holds for this $b(\gamma_1)$. 

Next, we confirm that the condition (ii) holds for sufficient large $\gamma$. 
For $\theta \in D_{\lambda_-(\gamma),b(\gamma_1),c,l}$, we evaluate
\begin{equation}\label{evoff}
\begin{split}
| \lambda_-(\gamma) f(x,\theta(x);\lambda_-(\gamma)) | 
&\leq   \left|\frac{   1 - a(x)(\theta(x) x^2 + 1)^2 - \frac{x}{2}(\theta(x) x^2+1)^2\frac{d}{dx}a(x)    }{x^2\sqrt{\gamma \left( 1-\frac{a(x)x^2}{\gamma} \right) \left( 1-a(x)(\theta(x) x^2+1)^2 \right) }} - \frac{2\lambda_-(\gamma)}{x} \right| \\
%&= \left|\frac{   1 - 2\theta(x) + xC_1(x,\theta(x)) - 2\lambda_-(\gamma) \sqrt{\gamma \left( 1-\frac{a(x)x^2}{\gamma} \right)} \sqrt{ \frac{1}{2} - 2\theta(x) + xC_2(x,\theta(x))}   }{x\sqrt{\gamma \left( 1-\frac{a(x)x^2}{\gamma} \right)\left( \frac{1}{2} - 2\theta(x) + xC_2(x,\theta(x)) \right)} }  \right| \\
&= \left|\frac{  C_3(x,\theta(x)) - 2\lambda_-(\gamma) \sqrt{\gamma \left( 1-\frac{a(x)x^2}{\gamma} \right)} \sqrt{ C_4(x,\theta(x))}   }{x\sqrt{\gamma \left( 1-\frac{a(x)x^2}{\gamma} \right)C_4(x,\theta(x))} }  \right| \\
%&=  \left|  \frac{  \left( C_3(x,\theta(x)) \right)^2 - 4\lambda_-(\gamma)^2 \gamma\left( 1-\frac{a(x)x^2}{\gamma} \right)  C_4(x,\theta(x))  }{x\sqrt{\gamma \left( 1-\frac{a(x)x^2}{\gamma} \right)C_4(x,\theta(x))} \left( C_3(x,\theta(x)) + 2\lambda_-(\gamma) \sqrt{\gamma\left( 1-\frac{a(x)x^2}{\gamma} \right)} \sqrt{ C_4(x,\theta(x))} \right)}  \right| \\
&\leq  
%\left|  \frac  {  \left\{ (\theta(x) - \lambda_-(\gamma)) - 1+2\lambda -xC_1(x,\theta(x)) + 2 \lambda_-(\gamma)^2 \gamma \right\} 4(\theta(x)-\lambda_-(\gamma))  }  {x\sqrt{\gamma \left( 1-\frac{a(x)x^2}{\gamma} \right)C_4(x,\theta(x)) } \left( C_3(x,\theta(x)) + 2\lambda_-(\gamma) \sqrt{\gamma\left( 1-\frac{a(x)x^2}{\gamma} \right)} \sqrt{ C_4(x,\theta(x))} \right)}\right|\\
%&+ \left|\frac{(2- 4\lambda_-(\gamma) + C_1(x,\theta(x)))C_1(x,\theta(x)) - 4\lambda_-(\gamma)^2\gamma  C_2(x,\theta(x)) +4\lambda_-(\gamma)^2 a(x)x^2 \left( \frac{1}{2} - 2\theta(x) + xC_2(x,\theta(x))\right)} {\sqrt{\gamma \left( 1-\frac{a(x)x^2}{\gamma} \right)C_4(x,\theta(x))}  \left( C_3(x,\theta(x)) + 2\lambda_-(\gamma) \sqrt{\gamma\left( 1-\frac{a(x)x^2}{\gamma} \right)} \sqrt{C_4(x,\theta(x))} \right)}\right|\\
%&\equiv 
\frac{F_1(x,\theta(x),\lambda_-(\gamma))|\theta(x)-\lambda_-(\gamma)| +xF_2(x,\theta(x),\lambda_-(\gamma))}{\sqrt{\gamma}xF_3(x,\theta,\lambda_-(\gamma))}
\end{split}
\end{equation}
where 
\begin{eqnarray}
& & x^2C_3(x,\theta)\equiv x^2 - 2\theta x^2 + x^3C_1(x,\theta) \nonumber \\
& & ~~~~~~~~~~~~~\equiv  1 - a(x)(\theta x^2 + 1)^2 - \frac{x}{2}(\theta x^2+1)^2\frac{d}{dx}a(x), \\
& & x^2C_4(x,\theta)\equiv \frac{x^2}{2} - 2\theta x^2 + x^3C_2(x,\theta)\equiv 1 - a(x)(\theta x^2 + 1)^2,\\
& & F_1(x,\theta,\lambda_-(\gamma))\equiv4\Big| (\theta - \lambda_-(\gamma)) - 1+2\lambda_-(\gamma) -xC_1(x,\theta) + 2 \lambda_-^2(\gamma) \gamma \Big|, \\
& & F_2(x,\theta,\lambda_-(\gamma))\equiv \Big|2(1- 2\lambda_-(\gamma))C_1(x,\theta) + C_1^2(x,\theta) - 4\lambda_-^2(\gamma)\gamma  C_2(x,\theta) \nonumber \\
& & ~~~~~~~~~~~~~~~~~~~+4\lambda_-^2(\gamma) a(x)x^2 \left( \frac{1}{2} - 2\theta + xC_2(x,\theta)\right)\Big|
\end{eqnarray}
and
\begin{equation}
\begin{split}
F_3(x,\theta,\lambda_-(\gamma))\equiv&\sqrt{\left( 1-\frac{a(x)x^2}{\gamma} \right)C_4(x,\theta(x))}  \\
& \times \left( C_3(x,\theta(x)) + 2\lambda_-(\gamma) \sqrt{\gamma\left( 1-\frac{a(x)x^2}{\gamma} \right)} \sqrt{C_4(x,\theta(x))} \right),
\end{split}
\end{equation}
respectively (see Appendix A for the details of the above evaluations). From equation (\ref{EXPofa}), $C_1(x,\theta)$, $C_2(x,\theta)$, $C_3(x,\theta)$ 
and $C_4(x,\theta)$ defined in the aboves are $C^\infty$ functions in any region. Here let us define
\begin{equation}
\tilde{D}_{\gamma_1,b(\gamma_1),c,l} \equiv \bigcup_{\lambda \in [0,\lambda_-(\gamma_1)]} \{ (x,\theta(x)) | x \in[0,l], \theta \in D_{\lambda,b(\gamma_1),c,l} \},
\end{equation}
where $b(\gamma_1)$ is introduced just before equation (\ref{bgamma1}). 
Note that $\tilde{D}_{\gamma_1,b(\gamma_1),c,l}$ is a closed compact subset and does not contain the singular line $\theta = \theta_S(x)$. Since (\ref{monoa}) 
and $1 - a(x)(\theta x^2 + 1)^2 > 0$ hold in $\tilde{D}_{\gamma_1,b(\gamma_1),c,l}$ except for $x = 0$, $x^2C_3(x,\theta)$ and 
$x^2C_4(x,\theta)$ can be zero only at $x = 0$ in $\tilde{D}_{\gamma_1,b(\gamma_1),c,l}$. On the other hand, 
$C_3(0,\theta) =  \frac{1}{2} - 2\theta$ and $C_4(0,\theta) = 1-2\theta$ do not vanish in $\tilde{D}_{\gamma_1,b(\gamma_1),c,l}$. 
Then, $C_3(x,\theta)>0$ and $C_4(x,\theta)>0$ hold in $\tilde{D}_{\gamma_1,b(\gamma_1),c,l}$. In addition, from definition (\ref{lambda}), 
$2\lambda_-(\gamma) \sqrt{\gamma}$ appeared in the above $F_1(x,\theta,\lambda_-(\gamma))$, $F_2(x,\theta,\lambda_-(\gamma))$ and $F_3(x,\theta,\lambda_-(\gamma))$ are strictly positive and bounded 
for $\gamma$ that $\lambda_-(\gamma)$ is in $[0,\lambda_-(\gamma_1)]$ and, using condition (\ref{eta_0}), we can see easily that $1-\frac{a(x)x^2}{\gamma}$ 
is also strictly positive and bounded for $\gamma$ that $\lambda_-(\gamma)$ is in $[0,\lambda_-(\gamma_1)]$. Thus, we conclude that there exist the strictly positive values 
$\nu_1$, $\nu_2$ and $ \nu_3$ defined by
\begin{equation}
\nu_1 \equiv \max_{\lambda_-(\gamma) \in [0,\lambda_-(\gamma_1)]}\Bigl\{ \max_{(x,\theta)\in \tilde{D}_{\gamma_1,b(\gamma_1),c,l}} F_1(x,\theta,\lambda_-(\gamma))\Bigr\},
\end{equation}\
\begin{equation}
\nu_2 \equiv  \max_{\lambda_-(\gamma) \in [0,\lambda_-(\gamma_1)]}\Bigl\{ \max_{(x,\theta)\in \tilde{D}_{\gamma_1,b(\gamma_1),c,l} } F_2(x,\theta,\lambda_-(\gamma))\Bigr\}\end{equation}
and
\begin{equation}
\nu_3 \equiv  \min_{\lambda_-(\gamma) \in [0,\lambda_-(\gamma_1)]}\Bigl\{ \min_{(x,\theta)\in \tilde{D}_{\gamma_1,b(\gamma_1),c,l}} F_3(x,\theta,\lambda_-(\gamma))\Bigr\},
\end{equation}
respectively. Using these values, for arbitrary $x \in[0,l]$, $\theta \in D_{\lambda_-(\gamma),b(\gamma_1),c,l}$ and $\lambda_-(\gamma) \in [0,\lambda_-(\gamma_1)]$, 
we see
\begin{equation}
\begin{split}
| \lambda_-(\gamma) f(x,\theta(x);\lambda_-(\gamma)) | &\leq \frac{\nu_1|\theta(x)-\lambda_-(\gamma)|+\nu_2 x}{\sqrt{\gamma}\nu_3 x}\\
&\leq \frac{\nu_1b(\gamma_1) x^c+\nu_2 x}{\sqrt{\gamma}\nu_3 x}.
\end{split}
\end{equation}
Then, we have
\begin{equation}
|T_{\lambda_-(\gamma)}(\theta)-\lambda_-(\gamma)| \leq \frac{\nu_1 b(\gamma_1)}{ (c+2)\sqrt{\gamma}\nu_3}x^c + \frac{\nu_2}{3\sqrt{\gamma}\nu_3}x.
\end{equation}
Here there exists $\gamma_2$ in $[\gamma_1,\infty)$ such that, for arbitrary $\gamma \in [\gamma_2,\infty)$,
\begin{equation}
\frac{\nu_1 b(\gamma_1)}{ (c+2)\sqrt{\gamma}\nu_3} + \frac{\nu_2}{3\sqrt{\gamma}\nu_3}l^{1-c} \leq b(\gamma_1)
\end{equation}
holds. Thus, for a $c \in (0,1)$ and arbitrary $\gamma \in [\gamma_2,\infty)$, we obtain
\begin{equation}
\begin{split}
|T_{\lambda_-(\gamma)}(\theta)-\lambda_-(\gamma)| &\leq \frac{\nu_1 b(\gamma_1)}{ (c+2)\sqrt{\gamma}\nu_3}x^c + \frac{\nu_2}{3\sqrt{\gamma}\nu_3}x\\
&\leq \frac{\nu_1 b(\gamma_1)}{ (c+2)\sqrt{\gamma}\nu_3}x^c + \frac{\nu_2}{3\sqrt{\gamma}\nu_3}l^{1-c}x^c\\
&\leq b(\gamma_1)x^c.
\end{split}
\end{equation}
This means that $T_{\lambda_-(\gamma)}$ maps $D_{\lambda_-(\gamma),b(\gamma_1),c,l}$ into itself, that is, 
the condition (ii) holds for a $c \in (0,1)$ and arbitrary $\gamma \in [\gamma_2,\infty)$. Therefore, from 
Theorem 3 and 4, $\theta_{n_0}(x;\gamma)$, which is the earliest of all future directed causal line emanating 
from the central singularity, exists in the range $[0,l]$. Finally, we will examine the condition (iii) for sufficient 
large $\gamma$. For arbitrary $\epsilon>0$, there exists $\gamma_3$ such that
\begin{equation}
\begin{split}
\theta_{AH}(l;\gamma) &= \frac{1}{l^2} \left(\frac{1}{\sqrt{a(l)}} \sqrt{1 - \frac{a(l)l^2}{\gamma}} - 1\right)\\
&> \frac{1}{l^2} \left(\frac{1}{\sqrt{a(l)}} - 1\right) - \epsilon \\
&\geq \theta_{S,min} - \epsilon
\end{split}
\end{equation}
holds for arbitrary $\gamma$ in $[\gamma_3,\infty)$. In the above, we used equations (\ref{locationofS}) 
and (\ref{locationofAH}) and definition (\ref{defthetasmin}) for $\theta_{S,min}$. Now we choose $\epsilon>0$ such that 
$\theta_{S,min} - \epsilon > \lambda_-(\gamma_1) + b(\gamma_1)l^c$ holds. Then, we have
\begin{equation}
\theta_{AH}(l;\gamma) > \theta_{S,min} - \epsilon>  \lambda_-(\gamma_1) + b(\gamma_1)l^c \geq \lambda_-(\gamma) + b(\gamma_1)l^c \geq \theta_{n_0}(l;\gamma)
\end{equation}
for any $\gamma$ in $[\gamma_0,\infty)$, where $\gamma_0 \equiv \max\{\gamma_1, \gamma_2, \gamma_3\}$.  

\end{proof}

Therefore, for all $\gamma \in [\gamma_0,\infty)$, $\theta_{n_0}(x;\gamma)$ arrives at the surface 
of the dust cloud before the event horizon appears there, that is, the central singularity is globally naked in this case. 

On the other hands, for $\gamma$ which is sufficiently close to $\eta$ defined by (\ref{eta_0}), 
we show that the central singularity is surrounded by the event horizon, that is, the central 
singularity is only locally naked.
\begin{lemma}(i) For any initial density distribution which is parameterized by equation (\ref{DPrho}) and satisfies 
$\eta \geq \gamma_{min}={\sqrt {11+5{\sqrt {5}}}}$, 
there exists $\gamma_1$ such that $\gamma_1 \to \infty$ for $a(l) \to 1$ and the central singularity is only locally naked for arbitrary 
$\gamma \in (\eta,\gamma_1]$. 

(ii) For any initial density distribution which is parameterized by equation (\ref{DPrho}) and satisfies $\eta < \gamma_{min}$, 
if there exists $x_0$ in $[0,l]$, which satisfies $\frac{\gamma_{min}}{\gamma_{min} + x_0^2} < a(x_0)$, 
then there exists $\gamma_2$ such that $\gamma_2 \to \infty$ for $a(l) \to 1$ and the central singularity is only 
locally naked for arbitrary $\gamma \in [\gamma_{min},\gamma_2]$.
\end{lemma}

\begin{proof}
We suppose that $\lambda$ satisfying equation (\ref{lambda}) exists. In this case, from Theorem 3, the central singularity is locally naked at least. Let us define $x_\eta$ as
\begin{equation}
a(x_\eta)x_\eta^2\equiv \max_{x \in [0,l]} a(x)x^2 = \eta.
\end{equation}
Here note that $x_\eta \neq 0$ because $a(0)$ is finite and $a(x)x^2>0$ except for $x = 0$. At $x = x_\eta$, the apparent horizon appears at
\begin{equation}\label{AHa}
\theta_{AH}(x_\eta;\gamma) = \frac{1}{x_\eta^2} \left(\frac{1}{\sqrt{a(x_\eta)}} \sqrt{1 - \frac{\eta}{\gamma}} - 1\right).
\end{equation}
If $\eta \geq \gamma_{min}$ holds, there exists $\gamma_1$ such that the right-hand side of this equation becomes negative 
for arbitrary $\gamma$ in $(\eta,\gamma_1]$. Additionally, since $\sqrt{1 - \frac{\eta}{\gamma}} < 1$ always holds, $\theta_{AH}(x_\eta;\gamma)$ would be negative if $a(x_\eta)$ were equal to $1$. This fact and $a(l) \leq a(x_\eta) \leq 1$ tell us that 
 $\gamma_1 \to \infty$ for $a(l) \to 1$. 
Since $\theta_{AH}(x_\eta;\gamma) < 0$ means that the apparent horizon exists at the earlier timeslice than the central singularity 
appears, null geodesics emanating from the central singularity can not arrive at future null infinity for arbitrary $\gamma$ in $(\eta,\gamma_1]$, 
that is, the central singularity is only locally naked.

On the other hand, for $\eta < \gamma_{min}$, $\gamma$ can not approach to $\eta$. But if there exists $x_0$ in $[0,l]$ that satisfies 
$\frac{\gamma_{min}}{\gamma_{min} + x_0^2} < a(x_0)$, then $\theta_{AH}(x_0;\gamma_{min}) < 0$ holds from equation (\ref{AHa}). 
Since $\theta_{AH}(x;\gamma)$ is continuous with respect to $\gamma$, there exists $\gamma_2$ such that $\theta_{AH}(x_0;\gamma) < 0$ 
holds for arbitrary $\gamma \in [\gamma_{min},\gamma_2]$, that is, the central singularity is only locally naked in these cases. 
In addition, since $\frac{\gamma_{min}}{\gamma_{min} + x_0^2} < 1$ and $a(l) \leq a(x_0) \leq 1$ always hold, we have $\gamma_2 \to \infty$ 
for $a(l) \to 1$.
\end{proof}

Furthermore, we can show the monotonicity of $\theta_{n_0}(x;\gamma)$ with respect to $\gamma$ at each $x$. 
Let us define $\theta(x;\gamma)$ as a solution to equation (\ref{nulleq}) for $\gamma$, which converges to $\lambda_-(\gamma)$ as $x \to 0$.

\begin{lemma}
For any initial density distribution parameterized as (\ref{DPrho}), $\theta(x;\gamma_s) >\theta(x;\gamma_l)$ 
holds for two different value of $\gamma$, $\gamma_s$ and $\gamma_l$ such as $\gamma_s < \gamma_l$, and all  
$x$ that $\theta(x;\gamma_s)$ exists. In particular, $\theta_{n_0}(x;\gamma)$ defined in Lemma 5 
is a monotonically decreasing function of $\gamma$ at each $x$.
\end{lemma}

\begin{proof}
We suppose that $\gamma_s < \gamma_l$ and $\theta(x;\gamma_s)$ exists in the range $[0,d_s)$. Now let us define
\begin{equation}
I \equiv \{ x | \theta(x;\gamma_s) > \theta(x;\gamma_l)\}.
\end{equation}
$I$ is the union of intervals and not empty because $\theta(x;\gamma_s)$ and $\theta(x;\gamma_l)$ are continuous 
and $ \theta(0;\gamma_s)=\lambda_-(\gamma_s) >\lambda_-(\gamma_l) = \theta(0;\gamma_l)$. We shall show that $[0,d_0) \subset I$ 
implies $d_0 \in I$ for arbitrary $d_0 < d_s$ in the following. It implies $I = [0,d_s)$ because any interval contained in $I$ 
must not be a closed proper subset in $[0,d_s)$ by definition.

Now we suppose that $[0,d_0) \subset I$ and $0 < x_1 < x_2 < d_0 < d_s$. Let us define 
\begin{equation}
f_0(x,\theta) \equiv  \frac{   1 - a(x)(\theta x^2 + 1)^2 - \frac{x}{2}(\theta x^2+1)^2\frac{d}{dx}a(x)    }{x\sqrt{  1-a(x)(\theta x^2+1)^2 }}.
\end{equation}
From (\ref{monoa}) and the fact that $f_0(x,\theta)$ is differentiable except for the singularity, it is positive and 
$\frac{\partial f_0(x,\theta)}{\partial\theta}$ is finite for arbitrary $x \in [x_1,d_0]$ and $\theta \in [\theta(x;\gamma_l),\theta(x;\gamma_s)]$. Then, from equation (\ref{nulleq}), we have
\begin{equation}
\begin{split}
&\theta(x_1;\gamma_s) - \theta(x_1;\gamma_l) \\
&= \theta(x_2;\gamma_s) - \theta(x_2;\gamma_l) + \int^{x_2}_{x_1}\frac{2}{x}(\theta(x;\gamma_s) - \theta(x;\gamma_l) )dx \\
&~~- \int^{x_2}_{x_1}\left(\frac{1}{\sqrt{\gamma_s-a(x)x^2}} - \frac{1}{\sqrt{\gamma_l-a(x)x^2}}\right)\frac{f_0(x,\theta(x;\gamma_s))}{x}dx \\
&~~- \int^{x_2}_{x_1}\Bigl(f_0(x,\theta(x;\gamma_s)) - f_0(x,\theta(x;\gamma_l))\Bigr)\frac{1}{x\sqrt{\gamma_l-a(x)x^2}}dx.\\
& <\theta(x_2;\gamma_s) - \theta(x_2;\gamma_l) + \int^{x_2}_{x_1}\frac{2}{x}(\theta(x;\gamma_s) - \theta(x;\gamma_l) )dx \\
&~~+ \int^{x_2}_{x_1}|f_0(x,\theta(x;\gamma_s)) - f_0(x,\theta(x;\gamma_l))|\frac{1}{x\sqrt{\gamma_l-a(x)x^2}}dx.\\
&\leq \theta(x_2;\gamma_s) - \theta(x_2;\gamma_l) + \int^{x_2}_{x_1}\frac{2}{x}(\theta(x;\gamma_s) - \theta(x;\gamma_l) )dx \\
&~~+ \int^{x_2}_{x_1}\sup_{\theta(x;\gamma_l) \leq \theta \leq \theta(x;\gamma_s)}\left|\frac{\partial f_0(x,\theta)}{\partial\theta}\right|(\theta(x;\gamma_s) - \theta(x;\gamma_l) )\frac{1}{x\sqrt{\gamma_l-a(x)x^2}}dx.\\
& =  \theta(x_2;\gamma_s) - \theta(x_2;\gamma_l) + \int^{x_2}_{x_1}F(x)(\theta(x;\gamma_s) - \theta(x;\gamma_l) )dx,
\end{split}
\end{equation}
where $F(x)$ is the positive function defined as
\begin{equation}
F(x) \equiv \frac{2}{x} + \sup_{\theta(x;\gamma_l) \leq \theta \leq \theta(x;\gamma_s)}\left|\frac{\partial f_0(x,\theta)}{\partial\theta}\right|\frac{1}{x\sqrt{\gamma_l-a(x)x^2}}.
\end{equation}
For the first inequality in the above, we used the fact that $\frac{1}{\sqrt{\gamma_s-a(x)x^2}} - \frac{1}{\sqrt{\gamma_l-a(x)x^2}}$ 
is positive because of $\gamma_s < \gamma_l$. Thus, we obtain
\begin{equation}
 \theta(x_2;\gamma_s) - \theta(x_2;\gamma_l) > (\theta(x_1;\gamma_s) - \theta(x_1;\gamma_l) )\exp \Bigl(-\int^{x_2}_{x_1}F(x)dx \Bigr).
\end{equation}
As $x_2 \to d_0$, this inequality becomes
\begin{equation}
\theta(d_0;\gamma_s) - \theta(d_0;\gamma_l) > (\theta(x_1;\gamma_s) - \theta(x_1;\gamma_l) )\exp \Bigl(-\int^{d_0}_{x_1}F(x)dx \Bigr) > 0
\end{equation}
because $F(x)$ is bounded in the range $[x_1,d_0]$ and we supposed $\theta(x_1;\gamma_s) > \theta(x_1;\gamma_l)$. This means $d_0 \in I$.
\end{proof}

Here let us define $N$ as the set of real number $\lambda_-(\gamma)$ such that $G_{\lambda_-(\gamma),d}$ contains more than one element for some $d$. Then, from Lemma 8, we have the following Corollary. 
\begin{corollary}
$N$ is countable.
\end{corollary}
\begin{proof}
We suppose $\gamma_s < \gamma_l$ again. From Lemma 8, $\theta(x;\gamma_s) > \theta(x;\gamma_l)$ holds for arbitrary 
$\theta(x;\gamma_s)$ and $\theta(x;\gamma_l)$ at arbitrary $x$ that 
$\theta(x;\gamma_s)$ exists. So the geodesics $\theta(x;\gamma_s)$ and $\theta(x;\gamma_l)$ do not intersect in the domain of $\theta(x;\gamma_s)$. 
This means that $G_{\lambda_-(\gamma_s),d}(x)\bigcap G_{\lambda_-(\gamma_l),d}(x)$ is empty set for arbitrary $d$ and $x$ in the domain of $\theta(x;\gamma_s)$.  
In addition, since equation ($\ref{nulleq}$) satisfies the Lifshitz condition on arbitrary compact set which does not 
contain the singularity, the elements in $G_{\lambda_-(\gamma),d}$ do not intersect each other in the region which does not contain the singularity. 
So if $G_{\lambda_-(\gamma),d}$ contains two different function $\theta_1(x;\gamma)$ and $\theta_2(x;\gamma)$ which satisfy 
$\theta_1(x_0;\gamma) < \theta_2(x_0;\gamma)$ for a $x_0$, arbitrary solution to equation ($\ref{nulleq}$), $\theta(x)$, which satisfy 
$\theta_1(x_0;\gamma) < \theta (x_0)< \theta_2(x_0;\gamma)$ must be contained in $G_{\lambda_-(\gamma),d}$. Thus, for non-zero $x$, $G_{\lambda_-(\gamma),d}(x)$ is alway an interval in $\mathbb{R}$ if $G_{\lambda_-(\gamma),d}$ contains more than one element.

Now we assume that $N$ is uncountable. From Theorem 3, a solution $\theta(x;\gamma(\lambda_{M}))$ exists in the range $[0,d_{M}]$, which $d_{M}$ is a positive number. From lemma 8, for arbitrary 
$\lambda$ satisfying $\lambda < \lambda_{M}$, all solution $\theta(x;\gamma(\lambda))$ also exists in $[0,d_{M}]$ because the region $\theta <\theta(x;\gamma(\lambda_{M})) $ does not contain the singularity at $\theta = \theta_S(x)$. Thus, $G_{\lambda,d_{M}}(d_{M})$ is interval for 
all $\lambda$ in $N$. Here we define $|G_{\lambda,d}(x)|$ as the Lebesgue measure of $G_{\lambda,d}(x)$.
$|G_{\lambda,d_{M}}(d_{M})|$ is non-zero for arbitrary $\lambda$ in $N$. 
We can evaluate the sum of $ |G_{\lambda,d_{M}}(d_{M})|$ for $\lambda$ in $N$ as
\begin{equation}\label{EofsumG}
\begin{split}
\sum_{\lambda \in N } \left|G_{\lambda,d_{M}}({d_{M}})\right| &\leq \left|\bigcup_{\lambda \in N }G_{\lambda,d_{M}}(d_{M})\right|\\
&\leq \left|[0,\theta_S(d_{M})]\right| = \theta_S(d_{M}).
\end{split}
\end{equation}
For the first inequality, we used the fact that $G_{\lambda,d_{M}}(d_{M})$ is interval for all $\lambda$ in $N$ and does not have a common part each other for different $\lambda$. 
For the second one, we used the line $\theta = \theta(x;\gamma(\lambda))$ does not enter the noncentral singularity in the range $[0,d_{M}]$ for all $\lambda$ in $N$ and the region 
$\theta < 0$ at $x = d_{M}$, which is not in the future of the central singularity. 
However, since the sum of uncountable infinite numbers of strictly positive real number must diverge, 
$\sum_{\lambda \in N} |G_{\lambda,d_{M}}(d_{M})|$ must diverge. 
It contradicts inequality (\ref{EofsumG}). Thus, $N$ is countable.
\end{proof}

Since the existence theorem is based on the fixed point theorem for contraction mapping in the four dimensional case \cite{Christodoulou:1984mz}, 
one could immediately see that the solution to the differential equation for null geodesic, which has certain initial value at central singularity, 
is unique. By contrast, in the five dimensional case, it is not necessary that the solution found in Thoerem 3 is unique because we use 
Schauder fixed-point theorem for the proof of the existence of the solution. However, this corollary guarantees that the solution which converges 
$\lambda_-(\gamma)$ as $x \to 0$ is unique for almost every $\gamma$ at least.

In lemma 8, we proved the monotonicity of the solutions to equation (\ref{nulleq}) with respect to $\gamma$. 
In addition, we can easily  show the monotonicity of $\theta_{AH}(x;\gamma)$ with 
respect to $\gamma$ at each $x$.
\begin{lemma}
For any initial density distribution parameterized as (\ref{DPrho}), $\theta_{AH}(x;\gamma)$ 
is a monotonically increasing function of $\gamma$ at each $x$.
\end{lemma}
\begin{proof}
It is obvious from equation (\ref{locationofAH}).
\end{proof}

From Lemma 6, 7, 8 and 9, we obtain the following theorem.
\begin{theorem}
(i) For any initial density distribution which is parameterized as (\ref{DPrho}) and satisfies 
$\eta \geq \gamma_{min}$, there exists $\gamma_C$ which satisfies $\eta < \gamma_C$ and $\gamma_C \to \infty$ for 
$a(l) \to 1$ such that (a) for arbitrary $\gamma \in (\gamma_C,\infty)$, $\theta_{n_0}(x;\gamma)$ defined in Lemma 5 
goes to future null infinity, that is, the central singularity is globally naked and weak CCC 
dose not hold, and (b) for all $\gamma \in (\eta,\gamma_C)$, the central singularity is 
only locally naked, that is, weak CCC holds and the outer region of the event horizon is regular.

(ii) For any initial density distribution which is parameterized as (\ref{DPrho}) and satisfies
$\eta < \gamma_{min}={\sqrt {11+5{\sqrt {5}}}}$, if there exists $x_0$ in $[0,l]$ such that 
$\frac{\gamma_{min}}{\gamma_{min} + x_0^2} < a(x_0)$ holds, then there exists 
$\gamma_C$ which satisfies $\eta < \gamma_C$ and $\gamma_C \to \infty$ for $a(l) \to 1$ 
such that the above (a) and (b) hold. Otherwise, there exist $\gamma_0$ satisfying $\eta < \gamma_0$, 
such that, for all $\gamma \in [\gamma_0,\infty)$, $\theta_{n_0}(x;\gamma)$ 
goes to future null infinity, that is, the central singularity is globally naked.
\end{theorem}
\begin{proof}
(i) Let us assume that the initial density distribution is parameterized as (\ref{DPrho}) and satisfies $\eta \geq \gamma_{min}$. Then, from Lemma 8 and 9, $\theta_{n_0}(x;\gamma)$ is a decreasing function of $\gamma$ and 
$\theta_{AH}(x;\gamma)$ is a continuous increasing function of $\gamma$ for each $x$. 
In addition, from Lemma 6, there exists $\gamma_0$ such that the solution $\theta_{n_0}(x;\gamma)$ 
can extend to $x = l$ and $\theta = \theta_{n_0}(x;\gamma)$ does not intersect with 
$\theta=\theta_{AH}(x;\gamma)$ for all $\gamma \in [\gamma_0, \infty]$, while from Lemma7, 
if $\eta \geq \gamma_{min}$ holds, there exists $\gamma_1$ such that $\gamma_1 \to \infty$ 
for $a(l) \to 1$ and $\theta = \theta_{n_0}(x;\gamma)$ intersects with 
$\theta=\theta_{AH}(x;\gamma)$ at somewhere in the dust cloud for arbitrary $\gamma \in (\eta,\gamma_1]$. Thus there exists $\gamma_C$ such that $\gamma_0 \geq \gamma_C \geq \gamma_1$ and 
$\theta = \theta_{n_0}(x;\gamma)$ does not intersect with $\theta=\theta_{AH}(x;\gamma)$ 
for all $\gamma \in (\gamma_C,\infty)$ and $\theta = \theta_{n_0}(x;\gamma)$ intersects 
with $\theta=\theta_{AH}(x;\gamma)$ at somewhere in the dust cloud for arbitrary 
$\gamma \in (\eta,\gamma_C)$. If $\theta = \theta_{n_0}(x;\gamma)$ does not intersect 
with $\theta=\theta_{AH}(x;\gamma)$, then $\theta = \theta_{n_0}(x;\gamma)$ can extend to future null infinity because the outer region of the $x = l$ surface is the Schwarzschild spacetime. Thus, in this case, the central singularity is globally naked. If $\theta = \theta_{n_0}(x;\gamma)$ 
intersects with $\theta=\theta_{AH}(x;\gamma)$ at somewhere in the dust cloud, 
then $\theta = \theta_{n_0}(x;\gamma)$ can not extend to future null infinity and will enter the singularity. 
This means that the central singularity is only locally naked and the outer region of the 
event horizon is regular because $\theta=\theta_{n_0}(x;\gamma)$ is the earliest line in all future 
directed causal lines emanating from the central singularity. 

(ii)  Let us assume that the initial density distribution is parameterized as (\ref{DPrho}) and satisfies $\eta < \gamma_{min}={\sqrt {11+5{\sqrt {5}}}}$. If there exists $x_0$ in $[0,l]$ such that $\frac{\gamma_{min}}{\gamma_{min} + x_0^2} < a(x_0)$ holds, in the same way as the proof of (i), we can show the former of (ii). On the other hand, if such $x_0$ does not exist in $[0,l]$, 
we can not use Lemma 7. Then all we could show in this regard is Lemma 6 only.
\end{proof}

\section{Conclusion and discussion}

In this paper, we analyzed five dimensional inhomogeneous spherically symmetric dust collapse. 
By virtue of Schauder fixed-point theorem, we proved the existence theorem of null geodesics 
in singular space-time. Moreover, by using it, we showed the necessary and sufficient condition 
for the singularity to be naked and saw the dependence of the globally nakedness of the central singularity on the initial density distribution.

In section 2, we fixed the initial energy distribution of dust so that the initial velocity of 
the shells is zero. This assumption is not critical for our method. Therefore, we can also 
discuss the nakedness of singularity without this assumption. 
To prove the existence of a null geodesic emanating from the central singularity in this general case, 
we have to find an appropriate domain such that the operator $T_{\lambda}$ maps its domain into itself. 
We expect that, for some class of energy distribution, $D_{\lambda,b,c,d}$ defined by (\ref{DD}) 
can be such domain for certain $b$, $c$ and $d$, and we will have almost similar discussion 
in this paper.

In specific dimensional spherically symmetric dust collapse in Lovelock gravity or, 
especially, nine dimensional spherically symmetric dust collapse in Einstein-Gauss-Bonnet 
gravity \cite{Maeda:2006pm}\footnote{In Lovelock gravity case, by employing the analysis 
in \cite{Maeda:2006pm}, it is easy to find that we cannot apply the Christodoulou theorem 
in $D=4k+1$ dimensional spacetime, where $k$ is the highest order of non vanishing Lovelock 
coefficients \cite{Ohashi:2011zza}.}, 
we can not use Christodoulou's method and discussion to show the existence of null geodesics 
emanating from central singularity because the singular term in differential equation for 
null geodesic has non simple function form. In contrast, our method may be used to examine 
the nakedness of singularity for the above cases because our existence theorem improved 
Christodoulou's method \cite{Christodoulou:1984mz}.

\section*{Acknowledgment}

We would like to thank Sumio Yamada and Hisashi Okamoto for fruitful discussions and comment. R. M. was supported by JSPS Grant-in-Aid for Scientific Research No. 22-995. S. O. was supported by JSPS Grant-in-Aid for Scientific Research No. 25-9997. 
T. S. is supported by Grant-Aid for Scientific Research from Ministry of Education, Science, Sports and Culture of Japan (Nos. 25610055 and 16K05344).


%\bibliographystyle{ptephy}
%\bibliography{sample}
%
% once the .bbl file has been generated then place the text in your article.

\appendix

\section{The details of the evaluation in equations (\ref{ineq0}), (\ref{ineq1}), (\ref{Tineq}) and (\ref{evoff})}

\subsection{equation (\ref{ineq0})}
\begin{equation}
\begin{split}
 | \lambda &f(x,\theta(x);\lambda) | =  \left|\frac{   1 - a(x)(\theta(x) x^2 + 1)^2 - \frac{x}{2}(\theta(x) x^2+1)^2\frac{d}{dx}a(x)    }{x^2\sqrt{\gamma \left( 1-\frac{a(x)x^2}{\gamma} \right) \left( 1-a(x)(\theta(x) x^2+1)^2 \right) }} - \frac{2\lambda}{x} \right| \\
& \leq \left|\frac{   1 - 2\theta(x) + O(x) - 2\lambda \sqrt{\gamma} \sqrt{ \frac{1}{2} - 2\theta(x) + O(x)}   }{\sqrt{\gamma}x \sqrt{ \frac{1}{2} - 2\theta(x) + O(x)} }  \right| \\
& =  \left|  \frac{  \left( 1 - 2\theta(x) + O(x) \right)^2 - 4\lambda^2 \gamma  \left( \frac{1}{2} - 2\theta(x) + O(x)\right)   }{\sqrt{\gamma}x \sqrt{ \frac{1}{2} - 2\theta(x) + O(x)} \left( 1 - 2\theta(x) + O(x) + 2\lambda \sqrt{\gamma} \sqrt{ \frac{1}{2} - 2\theta(x) + O(x)} \right)}  \right| \\
& =  \left|  \frac{  - 4(\theta(x) - \lambda) (1-2\lambda) +\left\{ - 2(\theta(x) - \lambda) + O(x) \right\}^2 +8\lambda^2 \gamma (\theta(x) - \lambda)  + O(x) }{\sqrt{\gamma}x \sqrt{ \frac{1}{2} - 2\theta(x) + O(x)} \left( 1 - 2\theta(x) + O(x) + 2\lambda \sqrt{\gamma} \sqrt{ \frac{1}{2} - 2\theta(x) + O(x)} \right)}  \right| \\
& \leq   \frac{  | -4(1-2\lambda) + 8\lambda^2 \gamma | |\theta(x) - \lambda| + O(|\theta(x)-\lambda|^2) + O(x) }{\sqrt{\gamma}x \sqrt{ \frac{1}{2} - 2\lambda + O(x^c) + O(x)} \left( 1 - 2\lambda + 2\lambda \sqrt{\gamma} \sqrt{ \frac{1}{2} - 2\lambda} + O(x^c) + O(x) \right)}  \\
& \leq  \frac{  | -4(1-2\lambda) + 8\lambda^2 \gamma | bx^{c-1} }{\sqrt{\gamma} \sqrt{ \frac{1}{2} - 2\lambda } \left( 1 - 2\lambda + 2\lambda \sqrt{\gamma} \sqrt{ \frac{1}{2} - 2\lambda} \right)}  + \frac{O(x^{2c}) + O(x)}{x}  \\
& \leq  \frac{ 2\lambda }{\sqrt{\gamma} \left( \frac{1}{2} - 2\lambda \right)^{\frac{3}{2}}}bx^{c-1} +O(1) + O(x^{2c-1}) .
\end{split}
\end{equation}
In the right-hand side of first inequality, we can choose functions $O$, which are independent of $\theta$, because $\theta(x)$ 
in $D_{\lambda,b,c,d}$ is uniformaly bounded by the constants $\lambda + bl^c$ and $\lambda - bl^c$. 
In the same way, we can choose functions $O$ which are independent of $\theta$.
\newpage
\subsection{equation (\ref{ineq1})}
\begin{equation}
\begin{split}
&|\lambda f(x,\theta_1(x);\lambda) - \lambda f(x,\theta_2(x);\lambda)| \\
&= \frac{1}{x^2\sqrt{\gamma -a(x)x^2 }} \left| \sqrt{g_1} - \sqrt{g_2} - \frac{x\frac{d}{dx}a(x)}{2a(x)} \left( \frac{1-g_1}{\sqrt{g_1}} - \frac{1-g_2}{\sqrt{g_2}}\right) \right| \\
&= \frac{1}{x^2\sqrt{\gamma -a(x)x^2 }} \left| 1 +  \frac{x\frac{d}{dx}a(x)}{2a(x)} \left( 1 + \frac{1}{\sqrt{g_1 g_2}}\right)\right| \left| \sqrt{g_1}-\sqrt{g_2}\right| \\
& \leq  \frac{1}{x^2\sqrt{\gamma -a(x)x^2 }} \left\{ 1+ \left| \frac{x\frac{d}{dx}a(x)}{2a(x)} \right| \left( 1 + \frac{1}{\sqrt{g_1g_2}} \right) \right\}\left( \frac{1}{\sqrt{g_1} + \sqrt{g_2}} \right) |g_1 - g_2| \\
& \leq \frac{|\theta_1(x)-\theta_2(x)|}{x^2\sqrt{\gamma -a(x)x^2 }}\left\{ 1- \frac{x\frac{d}{dx}a(x)}{2a(x)}  \left( 1 + \frac{1}{1-a(x)((\lambda+bx^c)x^2 +1)} \right) \right\}\frac{a(x) \{x^4(\lambda + bx^c) + x^2 \}}{\sqrt{(1-a(x)((\lambda+bx^c)x^2+1)^2}}\\
&= \frac{|\theta_1(x) - \theta_2(x)|}{x^2\sqrt{\gamma}}\left\{ 1 + \frac{x^2}{x^2-4\lambda x^2} \right\} \frac{x^2}{\sqrt{\frac{x^2}{2}-2\lambda x^2}} + B_1(x)x^{\delta - 1}|\theta_1(x) - \theta_2(x)| \\
&= \frac{1-2\lambda}{\sqrt{\gamma(\frac{1}{2}-2\lambda)^3}} x^{-1} |\theta_1(x) - \theta_2(x)| + B_1(x)x^{\delta - 1}|\theta_1(x) - \theta_2(x)|\\
&= \left(\frac{4\lambda}{1-4\lambda}x^{-1}  + B_1(x)x^{\delta - 1} \right) |\theta_1(x) - \theta_2(x)|,
\end{split}
\end{equation}
where $B_1(x)$ is introduced as the text. 
\subsection{equation (\ref{Tineq})}
\begin{equation*}
\begin{split}
|T_\lambda(\theta)(x)-&T_\lambda(\theta)(y)| \leq \left(\frac{1}{x^2} -\frac{1}{y^2} \right)\int^x_0 s^2|\lambda f(s,\theta(s);\lambda)| ds +\frac{1}{y^2} \int^y_x s^2|\lambda f(s,\theta(s);\lambda)| ds\\
&\leq \left(\frac{1}{x^2} -\frac{1}{y^2} \right)\int^x_0 s^2 \left\{\frac{ 2\lambda }{\sqrt{\gamma} \left( \frac{1}{2} - 2\lambda \right)^{\frac{3}{2}}}bs^{c-1}  + O(1) + O(s^{2c-1}) \right\}ds\\
&+\frac{1}{y^2}\int^y_x s^2 \left\{\frac{ 2\lambda }{\sqrt{\gamma} \left( \frac{1}{2} - 2\lambda \right)^{\frac{3}{2}}}bs^{c-1}  + O(1) + O(s^{2c-1})\right\}ds
\end{split}
\end{equation*}
\begin{equation}
\begin{split}
&= \left(\frac{1}{x^2} -\frac{1}{y^2} \right)h(x)x^{2+c} + \frac{1}{y^2}(h(y)y^{2+c}-h(x)x^{2+c})\\
&\leq h(x)|x^c-y^c|+y^c|h(y)-h(x)|+\frac{2h(x)}{y^2} |y^{2+c}-x^{2+c}|\\
&< h(x)|x^c-y^c|+y^c|h(y)-h(x)|+2h(x)y^c \left| \frac{y^2-x^2}{y^2} \right|\\
&\leq h(x)|x^c-y^c|+y^c|h(y)-h(x)|+2^{n+1}h(x)y^{c-\frac{1}{2^{n-1}}} \left| y^{\frac{1}{2^{n-1}}}-x^{\frac{1}{2^{n-1}}} \right|.
\end{split}
\end{equation}
\subsection{equation (\ref{evoff})}
\begin{equation}
\begin{split}
&| \lambda_-(\gamma) f(x,\theta(x);\lambda_-(\gamma)) | \leq   \left|\frac{   1 - a(x)(\theta(x) x^2 + 1)^2 - \frac{x}{2}(\theta(x) x^2+1)^2\frac{d}{dx}a(x)    }{x^2\sqrt{\gamma \left( 1-\frac{a(x)x^2}{\gamma} \right) \left( 1-a(x)(\theta(x) x^2+1)^2 \right) }} - \frac{2\lambda_-(\gamma)}{x} \right| \\
&= \left|\frac{   1 - 2\theta(x) + xC_1(x,\theta(x)) - 2\lambda_-(\gamma) \sqrt{\gamma \left( 1-\frac{a(x)x^2}{\gamma} \right)} \sqrt{ \frac{1}{2} - 2\theta(x) + xC_2(x,\theta(x))}   }{x\sqrt{\gamma \left( 1-\frac{a(x)x^2}{\gamma} \right)\left( \frac{1}{2} - 2\theta(x) + xC_2(x,\theta(x)) \right)} }  \right| \\
&\equiv \left|\frac{  C_3(x,\theta(x)) - 2\lambda_-(\gamma) \sqrt{\gamma \left( 1-\frac{a(x)x^2}{\gamma} \right)} \sqrt{ C_4(x,\theta(x))}   }{x\sqrt{\gamma \left( 1-\frac{a(x)x^2}{\gamma} \right)C_4(x,\theta(x))} }  \right| \\
&=  \left|  \frac{  \left( C_3(x,\theta(x)) \right)^2 - 4\lambda_-^2(\gamma) \gamma\left( 1-\frac{a(x)x^2}{\gamma} \right)  C_4(x,\theta(x))  }{x\sqrt{\gamma \left( 1-\frac{a(x)x^2}{\gamma} \right)C_4(x,\theta(x))} \left( C_3(x,\theta(x)) + 2\lambda_-(\gamma) \sqrt{\gamma\left( 1-\frac{a(x)x^2}{\gamma} \right)} \sqrt{ C_4(x,\theta(x))} \right)}  \right| \\
&\leq  
\left|  \frac  {  \left\{ (\theta(x) - \lambda_-(\gamma)) - 1+2\lambda_-(\gamma) -xC_1(x,\theta(x)) + 2 \lambda_-^2(\gamma) \gamma \right\} 4(\theta(x)-\lambda_-(\gamma))  }  {x\sqrt{\gamma \left( 1-\frac{a(x)x^2}{\gamma} \right)C_4(x,\theta(x)) } \left( C_3(x,\theta(x)) + 2\lambda_-(\gamma) \sqrt{\gamma\left( 1-\frac{a(x)x^2}{\gamma} \right)} \sqrt{ C_4(x,\theta(x))} \right)}\right|\\
&+ \left|\frac{(2- 4\lambda_-(\gamma) + C_1(x,\theta(x)))C_1(x,\theta(x)) - 4\lambda_-^2(\gamma) \{ \gamma  C_2(x,\theta(x)) +a(x)x^2 C_4(x,\theta(x)) \} } {\sqrt{\gamma \left( 1-\frac{a(x)x^2}{\gamma} \right)C_4(x,\theta(x))}  \left( C_3(x,\theta(x)) + 2\lambda_-(\gamma) \sqrt{\gamma\left( 1-\frac{a(x)x^2}{\gamma} \right)} \sqrt{C_4(x,\theta(x))} \right)}\right|\\
&=
\frac{F_1(x,\theta(x),\lambda_-(\gamma))|\theta(x)-\lambda_-(\gamma)| +xF_2(x,\theta(x),\lambda_-(\gamma))}{\sqrt{\gamma}xF_3(x,\theta(x),\lambda_-(\gamma))},
\end{split}
\end{equation}
where $C_1 \sim C_4$, $F_1$, $F_2$ and $F_3$ are defined as in the text.

\section{Four dimensional case}

In this appendix, we give an overview of Christodolou's paper \cite{Christodoulou:1984mz} 
which examined the global nakedness of singularity in four dimensional LTB spacetime, and 
see the difference between Christodoulou's and our discussions on the existence of null 
geodesics near the singuality. In the four dimensional case, after change of variables, 
the dimensionless differential equation for future directed null geodesic along the outer 
radial direction is given as
\begin{equation}
\frac{d\hat{\theta}}{d\hat{x}} + \frac{7\hat{\theta}}{\hat{x}}= \frac{7\hat{\lambda}}{\hat{x}} + \hat{\lambda}f_4(\hat{x},\hat{\theta};\hat{\lambda}),
\end{equation}
where, $\hat{\theta}$ and $\hat{x}$ are dimensionless coordinates, which correspond to $\theta$ 
and $x$ defined by (\ref{Dftheta}) and (\ref{Dfx}) respectively, $\hat{\lambda}$ is a certain 
constant and $f_4$ is a $C^\infty$ function. $\hat{\lambda}$ and $f_4$ are also the variables 
that correspond to $\lambda$ and $f$ defined in (\ref{lambda}) and (\ref{nulleq}) in five 
dimensional case, respectively. In order not to contain non-central singularity, $\hat{\theta}$ 
is restricted in the range, $ 0\leq \hat{\theta} < \sigma (\hat{x})$, where $\sigma$ is a certain 
function which satisfies $\sigma(\hat{x}) \geq \frac{\epsilon_4}{\hat{x}}$ for a positive
constant $\epsilon_4$. 

The formal solution to this differential equation is given by
\begin{equation}
\begin{split}
\hat{\theta}(\hat{x}) &= \lambda \left( 1 + \hat{x}\int_0^1 dv v^7 f_4(v\hat{x},\hat{\theta}(v\hat{x});\hat{\lambda})   \right) \\
& \equiv T_{4,\lambda}(\hat{\theta})(\hat{x}).
\end{split}
\end{equation}
Let us define
\begin{equation}
D_{\hat{d},\mu} \equiv \{\hat{\theta} | \hat{\theta} \in C^0[0,\hat{d}], 0 \leq \hat{\theta} \leq \mu\},
\end{equation}
where $\mu$ is a positive real number satisfying $\mu < \sigma(\hat{x})$ for all 
$\hat{x} \in [0,\hat{d}]$. $D_{\hat{d},\mu}$ becomes a subset of Banach space by uniform norm. 

After some discussion on the nature of $T_{4,\lambda}$, as with five dimensional case, 
we can conclude that $T_{4,\lambda}$ maps $D_{\hat{d},\mu}$ into itself for sufficient small $\hat{d}$. 
Furthermore, we obtain
\begin{eqnarray}\label{T4}
\|T_{4,\hat{\lambda}}(\hat{\theta}_1) - T_{4,\hat{\lambda}}(\hat{\theta}_2) \| &=& \sup_{0\leq \hat{x} \leq \hat{d}} \left| \hat{x} \int_0^1 v^7\hat{\lambda}\left\{  f_4(v\hat{x},\hat{\theta}_1(v\hat{x});\hat{\lambda}) -  f_4(v\hat{x},\hat{\theta}_2(v\hat{x});\hat{\lambda}) \right\}dv \right| \nonumber\\
&\leq& \frac{\hat{d}\hat{\lambda}\Delta}{8}\|\hat{\theta}_1-\hat{\theta}_2\|,
\end{eqnarray}
where $\Delta$ is defined as
\begin{equation}
\Delta \equiv  \sup_{0\leq \hat{x} \leq \hat{d}} 
\left\lbrace \sup_{0\leq \hat{\theta} \leq \mu} \left| \frac{\partial f_4}{\partial \hat{\theta}}(\hat{x},\hat{\theta};\hat{\lambda})\right| \right\rbrace.
\end{equation}
$\Delta$ is finite because $f_4$ is a $C^\infty$ function in $[0,\hat{d}] \times [0,\mu]$. 
Here we choose $\hat{d_0}$ so that it satisfies $\hat{d_0} \leq \hat{d}$ and
\begin{equation}
\hat{d_0} < \frac{8}{\hat{\lambda}\Delta},
\end{equation}
then $T_{4,\hat{\lambda}}$ becomes contraction mapping from $D_{\hat{d_0},\mu}$ into itself. Therefore, by the fixed-point 
theorem for contraction mapping \cite{banach1922operations}, we can conclude that $T_{4,\hat{\lambda}}$ has a unique fixed point, 
that is, a null geodesic emanating from the central singularity exists and the singularity is naked.

By contrast, in five dimensional case, what we can do is only to deform the differential equation 
for null geodesic near the central singularity like 
\begin{equation}\label{5geq}
\begin{split}
\frac{d\theta}{dx} + \frac{2\theta}{x}&= \frac{2\lambda}{x} + \frac{2\left( g(\theta;\gamma) + \theta - \lambda\right)}{x} + \lambda f_5(x,\theta;\gamma)\\
&\equiv \frac{2\lambda}{x} + \frac{2\lambda g_5(\theta;\gamma)}{x} + \lambda f_5(x,\theta;\gamma),
\end{split}
\end{equation}
where $g(\theta;\gamma)$ is defined by (\ref{Dg}), $ f_5$ is a function that $xf_5(x,\theta;\gamma)$ converges to $0$ as $x \to 0$ in the region $\theta < \theta_S(x)$, respectively. In four dimensional case, the right-hand side of the differential equation for 
null geodesic has a constant coefficient pole at first order only. However, in five dimensional 
case, the coefficient of the pole of the right-hand side of (\ref{5geq}) is a function with respect 
to $\theta$. Thus, the variable which corresponds to $\Delta$ in (\ref{T4}) is not finite in 
five dimension and we can not directly use the method employed for four dimensional case \cite{Christodoulou:1984mz}.

As above, we can apply the method in \cite{Christodoulou:1984mz} to the case that the geodesic 
equation has a constant coefficient pole at first order only. On the other hand, our method can 
be applied to the more general case that the geodesic equation can be deformed to the expression having 
a general pole at first order.

\end{document}